\documentclass[10pt]{amsart}
\usepackage{epsf,latexsym,amsmath,amssymb,amscd,datetime}
\usepackage{amsmath,amsthm,amssymb,enumerate,eucal,url,calligra,mathrsfs}

\usepackage{blkarray} 

\usepackage{tikz,scalerel}
\usetikzlibrary{arrows}

\usepackage{imakeidx}     

\usepackage[outdir=./]{epstopdf}

\usepackage{graphicx}
\usepackage{color}
\newenvironment{jfnote}{ \bgroup \color{blue} }{\egroup}

\newcommand{\red}{\color[rgb]{1.0,0.2,0.2}} 

\newcommand{\oldStuff}[1]{}

\newcommand{\Coord}{{\rm Coord}}
 
\newcommand{\og}{{\scriptscriptstyle \le}}
\newcommand{\Bg}{{\scalebox{1.0}{$\!\scriptscriptstyle /\!B$}}}

\DeclareMathOperator{\SHom}{\mathscr{H}\text{\kern -3pt {\calligra\large om}}\,}

\IfFileExists{my_xrefs}{\input my_xrefs}{}

\DeclareMathOperator{\ViSu}{VisSub}

\newcommand{\naturals}{{\mathbb N}}

\newcommand{\Eor}{E^{\mathrm{or}}}
\newcommand{\mec}[1]{{\bf #1}}	
\newcommand{\bec}[1]{{\boldsymbol #1}}	

\DeclareMathOperator{\Trace}{Trace}

\newcommand{\rhonew}{\rho^{\mathrm{new}}}
\newcommand{\specnew}{\Spec^{\mathrm{new}}}
\newcommand{\Specnew}{\Spec^{\mathrm{new}}}

\usepackage{mathrsfs}
\usepackage{amssymb}
\usepackage{dsfont}
\usepackage{verbatim}
\usepackage{url}

\newcommand{\Edir}{E^{\mathrm{dir}}}

\theoremstyle{plain}
\newtheorem{theorem}{Theorem}[section]
\newtheorem{lemma}[theorem]{Lemma}

\newtheorem{conjecture}[theorem]{Conjecture}

\theoremstyle{definition}
\newtheorem{definition}[theorem]{Definition}

\newtheorem{xca}{Exercise}[section]

\newcommand{\isom}{\simeq} 

\newcommand{\ignore}[1]{}

\newcommand{\reals}{{\mathbb R}}

\newcommand{\integers}{{\mathbb Z}}

\newcommand{\complex}{{\mathbb C}}

\newcommand\EE{\mathbb{E}}

\newcommand\II{\mathbb{I}}

\DeclareMathAlphabet{\mathcal}{OMS}{cmsy}{m}{n}

\newcommand\cC{\mathcal{C}}

\newcommand\cM{\mathcal{M}}

\newcommand\cP{\mathcal{P}}

\newcommand\cR{\mathcal{R}}
\newcommand\cS{\mathcal{S}}
\newcommand\cT{\mathcal{T}}

\newcommand\frakp{\mathfrak{p}}

\DeclareMathOperator{\Prob}{Prob}

\DeclareMathOperator{\VLG}{VLG}

\DeclareMathOperator{\Line}{Line}

\DeclareMathOperator{\SNBC}{SNBC}

\DeclareMathOperator{\snbc}{snbc}

\def\from{\colon}

\def\isom{\simeq}

\def\eqdef{\overset{\text{def}}{=}}

\DeclareMathOperator{\id}{id}

\DeclareMathOperator{\ord}{ord}

\DeclareMathOperator{\Spec}{Spec}

\def\implies{\Rightarrow}

\DeclareRobustCommand
  \rddots{\mathinner{\mkern1mu\raise\p@
    \vbox{\kern7\p@\hbox{.}}\mkern2mu
    \raise4\p@\hbox{.}\mkern2mu\raise7\p@\hbox{.}\mkern1mu}}

\newcommand\xhookrightarrow[2][]{\ext@arrow 0062{\hookrightarrowfill@}{#1}{#2}}
\def\hookrightarrowfill@{\arrowfill@\lhook\relbar\rightarrow}

\usepackage{relsize}
\usepackage{tikz}
\usetikzlibrary{matrix,arrows,decorations.pathmorphing}
\usepackage{tikz-cd}
\usetikzlibrary{cd}

\usepackage[pdftex,colorlinks,linkcolor=blue]{hyperref}

\newcommand{\myDeleteNote}[1]{{\red #1}}

\begin{document}

\title[Relativized Alon Conjecture V] 
{On the Relativized Alon Second Eigenvalue
Conjecture V: Proof of the Relativized Alon Conjecture for Regular Base
Graphs}

\author{Joel Friedman}
\address{Department of Computer Science,
        University of British Columbia, Vancouver, BC\ \ V6T 1Z4, CANADA}
\curraddr{}
\email{{\tt jf@cs.ubc.ca}}
\thanks{Research supported in part by an NSERC grant.}

\author{David Kohler}
\address{Department of Mathematics,
        University of British Columbia, Vancouver, BC\ \ V6T 1Z2, CANADA}
\curraddr{422 Richards St, Suite 170, Vancouver BC\ \  V6B 2Z4, CANADA}
\email{{David.kohler@a3.epfl.ch}}
\thanks{Research supported in part by an NSERC grant.}

%
\date{\today}

\subjclass[2010]{Primary 68R10}

\keywords{}

\begin{abstract}

This is the fifth in a series of articles devoted to showing that a typical
covering map of large degree to a fixed, regular graph has its new adjacency
eigenvalues within the bound conjectured by Alon for random regular graphs.

In this article we use the results of Articles~III and IV in this 
series to prove that if the base graph is regular, 
then as the degree, $n$, of the covering map tends to infinity,
some new adjacency eigenvalue has
absolute value outside the
Alon bound with probability bounded by $O(1/n)$.
In addition, we give upper and lower bounds on this probability
that are tight to within a multiplicative constant times the
degree of the covering map.
These bounds depend on two positive integers,
the {\em algebraic power} (which can also be $+\infty$)
and the {\em tangle power} of the model
of random covering map.

We conjecture that the algebraic power of the models we study is
always $+\infty$, and in Article~VI we prove this when the base
graph is regular and {\em Ramanujan}.  
When the algebraic power of the model is $+\infty$, then
the results in this article imply
stronger results, such as (1) the upper
and lower bounds mentioned above are matching to within a 
multiplicative constant, and
(2) with probability smaller than any negative power of the
degree, the some new eigenvalue fails to be within the Alon bound
only if the covering map contains one of finitely many ``tangles''
as a subgraph (and this event has low probability).

\end{abstract}

\maketitle
\setcounter{tocdepth}{3}
\tableofcontents

\newcommand{\sePrelimProofs}{17}

\section{Introduction}

The is the fifth article in a series of six articles whose goal is to prove
a relativization of
Alon's Second Eigenvalue Conjecture,
formulated in \cite{friedman_relative},
for any {\em base graph}, $B$, that is regular;
a proof of this theorem appears in our
preprint \cite{friedman_kohler}.
This series of six articles represents a
``factorization'' of the proof in \cite{friedman_kohler}
into many independent parts.
This series of articles includes some original work beyond 
that required to merely factor \cite{friedman_kohler}:
this series of articles has some simplifications and generalizations
of \cite{friedman_kohler} and of
\cite{friedman_random_graphs,friedman_alon} (on which much of
\cite{friedman_kohler} is based).
As such, the tools we develop in this series of articles will likely make it
easier to generalize these results to related questions about the
adjacency matrix eigenvalues of families
of random graphs.
Furthermore we close a gap in \cite{friedman_alon}
between the upper and lower bound
on the probability of having eigenvalues outside the Alon bound

In this article we complete the proof of the first main theorem in 
this series of articles.
This theorem shows that the Relativized Alon Conjecture holds for
algebraic models of coverings over any $d$-regular
base graph; more precisely, it shows that the probability of having
{\em non-Alon} new eigenvalues---meaning new eigenvalues
larger than $2(d-1)^{1/2}+\epsilon$ for a fixed $\epsilon>0$---for 
a random cover of degree $n$ is bounded
above by a function of 
order $1/n$.  This proof is given in Section~\ref{se_proof_first}.
Much of this article is devoted to proving a much stronger
theorem that results from our trace methods.  Let us describe this
theorem in rough terms.

For each graph $B$, we formulate models of random covering maps
that we call {\em our basic models}; these are based on the models
in \cite{friedman_alon}.  All these models turn out to be
{\em algebraic}, meaning that they satisfy
a set of conditions that allow us to apply
our trace methods.  To any algebraic model we 
associate an integer,
its {\em tangle power}, $\tau_{\rm tang}$; this is relatively easy
to estimate and was determined exactly in \cite{friedman_alon} when
$B$ is a bouquet of whole-loops or of half-loops (and therefore $B$ has
only one vertex).
To any algebraic model we also associate 
its {\em algebraic power}, $\tau_{\rm alg}$, which is either a positive
integer or $+\infty$; determining $\tau_{\rm alg}$ is much more
difficult in practice and relies on computing certain asymptotic
expansions involving the expected values
of certain traces of matrices (the {\em Hashimoto} or
{\em non-backtracking} matrix) associated to the covering graph.
In this paper we show that when the base graph is regular,
then the probability that a random covering
graph having a new eigenvalue outside of the Alon bound is
bounded above proportional to
$n^{-\tau_1}$ and below proportional to $n^{-\tau_2}$, where
$$
\tau_1 = \min(\tau_{\rm tang},\tau_{\rm alg}), \quad
\tau_2 = \min(\tau_{\rm tang},\tau_{\rm alg}+1).
$$
This is the second main theorem in this article, although there is
a more remarkable result that follows from the proof of this
theorem.
Namely, whenever $B$ is regular and $\tau_{\rm alg}=+\infty$, 
then the probability of having a non-Alon new eigenvalue 
is $O(n^{-j})$ for any $j$, provided that we discard those graphs
that contain one of finitely many tangles (this finite number
depends on $j$ and the $\epsilon>0$ above in defining non-Alon
eigenvalues); furthermore, the probability that such
tangles exist is $O(n^{-\tau_{\rm tang}})$ (where the constant
in the $O()$ depends on $j$ and $\epsilon$).
We conjecture that $\tau_{\rm alg}=+\infty$ for all graphs, $B$, and
in Article~VI (i.e., the sixth article in this series)
we will prove this when $B$ is regular and {\em Ramanujan};
in this case, then $\tau_1=\tau_2=\tau_{\rm tang}$.

The rest of this paper is organized as follows.
In Section~\ref{se_defs_review} we review the main definitions that
we will use in this series of papers; for more details and 
motivation regarding these definitions, we refer the reader to 
Article~I.

In Section~\ref{se_main_results} we state the two main results of this
article: first, that the Relativized Alon Bound holds for all regular
base graphs, and second, that the more precise bounds 
involving $\tau_1,\tau_2$ above hold;
we also include some conjectures regarding 
$\tau_{\rm alg}$ and discuss the consequences of
$\tau_{\rm alg}=+\infty$.
In Section~\ref{se_main_results} we also recall the main results of
Articles~III and IV which we need in this article.

Sections~\ref{se_ihara}--\ref{se_remaining_proofs} are devoted to
proving the main theorems in Section~\ref{se_main_results}.
In Section~\ref{se_ihara} we prove what is sometimes called 
Ihara's Determinantal Formula, which for $d$-regular graphs gives a
precise description of adjacency matrix eigenvalues in terms of those of
its 
{\em Hashimoto matrix} (also called the {\em non-backtracking matrix})
of the graph.
In Section~\ref{se_ours_algebraic} we prove that our {\em basic models}
of covering graphs of a given base graph satisfy the ``algebraic''
properties we need in Articles~II and III.
Both Sections~\ref{se_ihara} and Section~\ref{se_ours_algebraic} can
be viewed as ``loose ends'' from Article~I, and are independent of the
rest of this article.
In Section~\ref{se_proof_first} we prove the Relativized Alon
Conjecture for regular base graphs.
In Section~\ref{se_fundamental_subgr} we use the methods of
Friedman-Tillich \cite{friedman_tillich_generalized} to show that
the existence of certain {\em tangles} in any covering graph of sufficiently
large degree
implies the existence of a new eigenvalue outside of the Alon bound.
In Section~\ref{se_remaining_proofs} we complete the proof of the
second main theorem in this article.

In Section~\ref{se_markov_bounds} we make an observation,
apparently new as of \cite{friedman_kohler}, that applies to
trace methods
for random regular graphs, such as
\cite{friedman_random_graphs,friedman_relative,linial_puder,puder},
that prove a high probability new adjacency spectral
bound that is strictly greater than the Alon bound: namely, these
bounds can be improved by
the analogous trace methods applied to Hashimoto (i.e., non-backtracking)
new eigenvalues, and then converting these bounds back to
adjacency matrix bounds.
Section~\ref{se_markov_bounds} is independent of the rest of the
article, beyond some of the terminology in Section~\ref{se_defs_review}.

\section{Review of the Main Definitions}
\label{se_defs_review}

We refer the reader to Article~I for the definitions used in this article,
the motivation of such definitions, and an appendix there that lists all the
definitions and notation.
In this section we briefly review these definitions and notation. 

\subsection{Basic Notation and Conventions}
\label{su_very_basic}

We use $\reals,\complex,\integers,\naturals$
to denote, respectively, the
the real numbers, the complex numbers, the integers, and positive
integers or
natural numbers;
we use $\integers_{\ge 0}$ ($\reals_{>0}$, etc.)
to denote the set of non-negative
integers (of positive real numbers, etc.).
We denote $\{1,\ldots,n\}$ by $[n]$.

If $A$ is a set, we use $\naturals^A$ to denote the set of
maps $A \to \naturals$; we will refers to its elements as
{\em vectors}, denoted in bold face letters, e.g., $\mec k\in \naturals^A$
or $\mec k\from A\to\naturals$; we denote its {\em component}
in the regular face equivalents, i.e., for $a\in A$,
we use $k(a)\in\naturals$ to denote
the $a$-component of $\mec k$.
As usual, $\naturals^n$ denotes $\naturals^{[n]}=\naturals^{\{1,\ldots,n\}}$.
We use similar conventions for $\naturals$ replaced by $\reals$,
$\complex$, etc.

If $A$ is a set, then $\# A$ denotes the cardinality of $A$.
We often denote a set with all capital letters, and its cardinality
in lower case letters; for example,
when we define
$\SNBC(G,k)$, we will write
$\snbc(G,k)$ for $\#\SNBC(G,k)$.

If $A'\subset A$ are sets, then $\II_{A'}\from A\to\{0,1\}$ (with $A$
understood) denotes
the characteristic function of $A'$, i.e., $\II_{A'}(a)$ is $1$ if
$a\in A'$ and otherwise is $0$;
we also write $\II_{A'}$ (with $A$ understood) to mean $\II_{A'\cap A}$
when $A'$ is not necessarily a subset of $A$.

All probability spaces are finite; hence a probability space
is a pair $\cP=(\Omega,P)$ where $\Omega$ is a finite set and
$P\from \Omega\to\reals_{>0}$ with $\sum_{\omega\in\Omega}P(\omega)=1$;
hence an {\em event} is any subset of $\Omega$.
We emphasize that $\omega\in\Omega$ implies that $P(\omega)>0$ with
strict inequality; we refer to the elements of $\Omega$ as
the atoms of the probability space.
We use $\cP$ and $\Omega$ interchangeably when $P$ is
understood and confusion is unlikely.

A {\em complex-valued random variable} on $\cP$ or $\Omega$
is a function $f\from\Omega\to\complex$, and similarly for real-,
integer-, and natural-valued random variable; we denote its
$\cP$-expected value by
$$
\EE_{\omega\in\Omega}[f(\omega)]=\sum_{\omega\in\Omega}f(\omega)P(\omega).
$$
If $\Omega'\subset\Omega$ we denote the probability of $\Omega'$ by
$$
\Prob_{\cP}[\Omega']=\sum_{\omega\in\Omega'}P(\omega')
=
\EE_{\omega\in\Omega}[\II_{\Omega'}(\omega)].
$$
At times we write $\Prob_{\cP}[\Omega']$ where $\Omega'$ is
not a subset of $\Omega$, by which we mean
$\Prob_{\cP}[\Omega'\cap\Omega]$.

\subsection{Graphs, Our Basic Models, Walks}

A {\em directed graph},
or simply a {\em digraph},
is a tuple $G=(V_G,\Edir_G,h_G,t_G)$ consisting of sets
$V_G$ and $\Edir_G$ (of {\em vertices} and {\em directed edges}) and maps
$h_G,t_G$ ({\em heads}
and {\em tails}) $\Edir_G\to V_G$.
Therefore our digraphs can have multiple edges and
self-loops (i.e., $e\in\Edir_G$ with $h_G(e)=t_G(e)$).
A {\em graph} is a tuple $G=(V_G,\Edir_G,h_G,t_G,\iota_G)$
where $(V_G,\Edir_G,h_G,t_G)$ is a digraph and
$\iota_G\from \Edir_G\to \Edir_G$ is an involution with
$t_G\iota_G=h_G$;
the {\em edge set} of $G$, denoted $E_G$, is the
set of orbits of $\iota_G$, which (notation aside)
can be identified with $\Edir_G/\iota_G$,
the set of equivalence classes of
$\Edir_G$ modulo $\iota_G$;
if $\{e\}\in E_G$ is a singleton, then necessarily $e$ is a self-loop
with $\iota_G e =e $, and
we call $e$ a {\em half-loop}; other elements of $E_G$ are sets
$\{e,\iota_G e\}$ of size two, i.e., with $e\ne\iota_G e$, and for such $e$
we say that $e$ (or, at times, $\{e,\iota_G e\}$)
is a {\em whole-loop} if
$h_G e=t_G e$ (otherwise $e$ has distinct endpoints).

Hence these definitions allow our graphs to have multiple edges and 
two types of self-loops---whole-loops
and half-loops---as in
\cite{friedman_geometric_aspects,friedman_alon}.
The {\em indegree} and {\em outdegree} of a vertex in a digraph is
the number of edges whose tail, respectively whose head, is the vertex;
the {\em degree} of a vertex in a graph is its indegree (which equals
its outdegree) in the underlying digraph; 
therefore a whole-loop about a vertex contributes $2$
to its degree, whereas a half-loop contributes $1$.

An {\em orientation} of a graph, $G$, is a choice $\Eor_G\subset\Edir_G$
of $\iota_G$ representatives; i.e., $\Eor_G$ contains every half-loop, $e$,
and one element of each two-element set $\{e,\iota_G e\}$.

A {\em morphism $\pi\from G\to H$} of directed graphs is a pair
$\pi=(\pi_V,\pi_E)$ where $\pi_V\from V_G\to V_H$ and
$\pi_E\from \Edir_G\to\Edir_H$ are maps that intertwine the heads maps
and the tails maps of $G,H$ in the evident fashion;
such a morphism is {\em covering} (respectively, {\em \'etale},
elsewhere called an {\em immersion}) if for each $v\in V_G$,
$\pi_E$ maps those directed edges whose head is $v$ bijectively
(respectively, injectively) to those whose head is $\pi_V(v)$,
and the same with tail replacing head.
If $G,H$ are graphs, then a morphism $\pi\from G\to H$ is a morphism
of underlying directed graphs where $\pi_E\iota_G=\iota_H\pi_E$;
$\pi$ is called {\em covering} or {\em \'etale} if it is so as a morphism
of underlying directed graphs.
We use the words {\em morphism} and {\em map} interchangeably.

A walk in a graph or digraph, $G$, is an alternating sequence
$w=(v_0,e_1,\ldots,e_k,v_k)$ of vertices and directed edges
with $t_Ge_i=v_{i-1}$ and $h_Ge_i=v_i$ for $i\in[k]$;
$w$ is {\em closed} if $v_k=v_0$;
if $G$ is a graph,
$w$ is {\em non-backtracking}, or simply {\em NB},
if $\iota_Ge_i\ne e_{i+1}$
for $i\in[k-1]$, and {\em strictly 
non-backtracking closed}, or simply {\em SNBC},
if it is closed, non-backtracking, and 
$\iota_G e_k\ne e_1$.
The {\em visited subgraph} of a walk, $w$, in a graph $G$, denoted
$\ViSu_G(w)$ or simply
$\ViSu(w)$, is the smallest subgraph of $G$ containing all the vertices
and directed edges of $w$;
$\ViSu_G(w)$ generally depends on $G$, i.e., $\ViSu_G(w)$ cannot be inferred
from the sequence $v_0,e_1,\ldots,e_k,v_k$ alone without knowing
$\iota_G$.

The adjacency matrix, $A_G$,
of a graph or digraph, $G$, is defined as usual (its $(v_1,v_2)$-entry
is the number of directed edges from $v_1$ to $v_2$);
if $G$ is a graph on $n$ vertices, 
then $A_G$ is symmetric and we order its eigenvalues (counted with
multiplicities) and denote them
$$
\lambda_1(G)\ge \cdots \ge \lambda_n(G).
$$
If $G$ is a graph, its
Hashimoto matrix (also called the non-backtracking matrix), $H_G$,
is the adjacency matrix of the {\em oriented line graph} of $G$,
$\Line(G)$,
whose vertices are $\Edir_G$ and whose directed edges
are the subset of $\Edir_G\times\Edir_G$ consisting of pairs $(e_1,e_2)$
such that $e_1,e_2$ form the
directed edges of a non-backtracking walk (of length two) in $G$
(the tail of $(e_1,e_2)$ is $e_1$, and its head $e_2$);
therefore $H_G$
is the square matrix indexed on $\Edir_G$, whose $(e_1,e_2)$ entry
is $1$ or $0$ according to, respectively, whether or not
$e_1,e_2$ form a non-backtracking walk
(i.e., $h_G e_1=t_G e_2$ and $\iota_G e_1\ne e_2$).
We use $\mu_1(G)$ to denote the Perron-Frobenius eigenvalue of 
$H_G$, and use $\mu_i(G)$ with $1<i\le \#\Edir_G$ to denote the
other eigenvalues of $H_G$ (which are generally complex-valued)
in any order.

If $B,G$ are both digraphs,
we say that $G$ is a {\em coordinatized graph over $B$
of degree $n$}
if
\begin{equation}\label{eq_coord_def}
V_G=V_B\times [n], \quad\Edir_G=\Edir_B\times[n], \quad
t_G(e,i)=(t_B e,i),\quad
h_G(e,i)=(h_Be,\sigma(e)i)
\end{equation} 
for some map
$\sigma\from\Edir_B\to\cS_n$, where $\cS_n$ is the group
of permutations on $[n]$; we call $\sigma$ (which is uniquely determined by
\eqref{eq_coord_def}) {\em the permutation assignment
associated to $G$}.
[Any such $G$ comes with a map $G\to B$ given by 
``projection to the first component of
the pair,'' and this map is a covering map of degree $n$.]
If $B,G$ are graphs, we say that a graph $G$ is a 
{\em coordinatized graph over $B$
of degree $n$} if \eqref{eq_coord_def} holds and also
\begin{equation}\label{eq_coord_def_graph}
\iota_G(e,i) = \bigl( \iota_B e,\sigma(e)i \bigr) ,
\end{equation} 
which implies that 
\begin{equation}\label{eq_sigma_iota_B}
(e,i)=\iota_G\iota_G(e,i) = \bigl( e, \sigma(\iota_B e)\sigma(e)i \bigr)
\quad\forall e\in\Edir_B,\ i\in[n],
\end{equation}
and hence $\sigma(\iota_B e)=\sigma(e)^{-1}$;
we use $\Coord_n(B)$ to denote the set of all coordinatized covers
of a graph, $B$, of degree $n$.

The {\em order} of a graph, $G$, is $\ord(G)\eqdef (\#E_G)-(\#V_G)$.
Note that a half-loop and a whole-loop each contribute $1$ to 
$\#E_G$ and to the order of $G$.
The {\em Euler characteristic} of a graph, $G$, is
$\chi(G)\eqdef (\# V_G) - (\#\Edir_G)/2$.
Hence $\ord(G)\ge -\chi(G)$, with equality iff $G$ has no half-loops.

If $w$ is a walk in any $G\in\Coord_n(B)$, then one easily
sees that $\ViSu_G(w)$ can be inferred
from $B$ and $w$ alone.

If $B$ is a graph without half-loops, then the {\em permutation model over
$B$} refers to the probability spaces $\{\cC_n(B)\}_{n\in\naturals}$ where
the atoms of $\cC_n(B)$ are coordinatized coverings of degree $n$
over $B$ chosen with the uniform distribution.
More generally, a {\em model} over a graph, $B$, is a collection of
probability spaces, $\{\cC_n(B)\}_{n\in N}$, 
defined for $n\in N$ where $N\subset\naturals$ is an
infinite subset, and where the atoms of each $\cC_n(B)$ are elements
of $\Coord_n(B)$.
There are a number of models related to the permutation model,
which are generalizations of the models of \cite{friedman_alon},
that we call {\em our basic models} and are defined in Article~I;
let us give a rough description.

All of {\em our basic models} are {\em edge independent}, meaning that
for any orientation $\Eor_B\subset\Edir_B$, the values of 
the permutation assignment, $\sigma$, on $\Eor_B$ are independent
of one another (of course, $\sigma(\iota_G e)=(\sigma(e))^{-1}$,
so $\sigma$ is determined by its values on any orientation
$\Eor_B$); for edge independent models, it suffices to specify
the ($\cS_n$-valued)
random variable $\sigma(e)$ for each $e$ in $\Eor_B$ or $\Edir_B$.
The permutation model can be alternatively described as the 
edge independent model that assigns a uniformly chosen permutation
to each $e\in\Edir_B$ (which requires $B$ to have no half-loops);
the {\em full cycle} (or simply {\em cyclic}) model is the same, except
that if $e$ is a whole-loop then $\sigma(e)$ is chosen uniformly
among all permutations whose cyclic structure consists of a single
$n$-cycle.
If $B$ has half-loops, then we restrict $\cC_n(B)$ either to $n$ even
or $n$ odd and for each half-loop $e\in\Edir_B$ we
choose $\sigma(e)$ as follows: if $n$ is even we choose 
$\sigma(e)$ uniformly among all perfect matchings,
i.e., involutions (maps equal to their inverse) with no fixed points;
if $n$ is odd then we choose $\sigma(e)$ uniformly among
all {\em nearly perfect matchings}, meaning involutions with one
fixed point.
We combine terms when $B$ has half-loops: for example,
the term {\em full cycle-involution} (or simply {\em cyclic-involution})
{\em model of odd degree over $B$} refers
to the model where the degree, $n$, is odd,
where $\sigma(e)$ follows the full cycle rule when $e$ is
not a half-loop, and where $\sigma(e)$ is a near perfect matching
when $e$ is a half-loop;
similarly for the {\em full cycle-involution} (or simply 
{\em cyclic-involution})
{\em model of even degree}
and the {\em permutation-involution model of even degree}
or {\em of odd degree}.

If $B$ is a graph, then a model, $\{\cC_n(B)\}_{n\in N}$, over $B$
may well have $N\ne \naturals$ (e.g., our basic models above when
$B$ has half-loops); in this case many formulas involving
the variable $n$ are only defined for $n\in N$.  For brevity, we
often do not explicitly write $n\in N$ in such formulas; 
for example we usually write
$$
\lim_{n\to\infty} \quad\mbox{to abbreviate}\quad
\lim_{n\in N,\ n\to\infty} \ .
$$
Also we often write simply $\cC_n(B)$ or $\{\cC_n(B)\}$ for
$\{\cC_n(B)\}_{n\in N}$ if confusion is unlikely to occur.

A graph is {\em pruned} if all its vertices are of degree at least
two (this differs from the more standard definition of {\em pruned} 
meaning that there are
no leaves).  If $w$ is any SNBC walk in a graph, $G$, then
we easily see that
$\ViSu_G(w)$ is necessarily pruned: i.e., any of its vertices must be
incident upon a whole-loop or two distinct edges
[note that a walk of length $k=1$ about a half-loop, $(v_0,e_1,v_1)$, by
definition, is not SNBC since $\iota_G e_k=e_1$].
It easily follows that $\ViSu_G(w)$ is contained in the graph
obtained from $G$ by repeatedly ``pruning any leaves''
(i.e., discarding any vertex of degree one and its incident edge)
from $G$.
Since our trace methods only concern (Hashimoto matrices and)
SNBC walks, it suffices to work with models $\cC_n(B)$ where
$B$ is pruned.
It is not hard to see that if $B$ is pruned and connected,
then $\ord(B)=0$ iff $B$ is a cycle,
and $\mu_1(B)>1$ iff $\chi(B)<0$;
this is formally proven in Article~III (Lemma~6.4).
Our theorems are not usually interesting unless $\mu_1(B)>\mu_1^{1/2}(B)$,
so we tend to restrict our main theorems
to the case $\mu_1(B)>1$ or, equivalently,
$\chi(B)<0$; some of our techniques work without these restrictions.

\subsection{Asymptotic Expansions}
\label{su_asymptotic_expansions}


A function $f\from\naturals\to\complex$ is a {\em polyexponential} if
it is a sum of functions $p(k)\mu^k$, where $p$ is a polynomial
and $\mu\in\complex$, with the convention
that for $\mu=0$ we understand $p(k)\mu^k$ to mean
any function that vanishes for sufficiently large $k$\footnote{
  This convention is used because then for any fixed matrix, $M$,
  any entry of $M^k$, as a function of $k$, is a polyexponential
  function of $k$; more specifically, the $\mu=0$ convention
  is due to the fact that a Jordan block of eigenvalue $0$ is
  nilpotent.
  }; we refer to the $\mu$
needed to express $f$ as the {\em exponents} or {\em bases} of $f$.
A function $f\from\naturals\to\complex$ is {\em of growth $\rho$}
for a $\rho\in\reals$ if $|f(k)|=o(1)(\rho+\epsilon)^k$ for any $\epsilon>0$.
A function $f\from\naturals\to\complex$ is $(B,\nu)$-bounded if it
is the sum of a function of growth $\nu$ plus a polyexponential function
whose bases are bounded by $\mu_1(B)$ (the Perron-Frobenius eigenvalue
of $H_B$); the {\em larger bases} of $f$ (with respect to $\nu$) are
those bases of the polyexponential function that are larger in
absolute value than $\nu$.
Moreover, such an $f$ is called {\em $(B,\nu)$-Ramanujan} if its
larger bases are all eigenvalues of $H_B$.

We say that a function $f=f(k,n)$ taking some subset of $\naturals^2$ to
$\complex$ has a 
{\em $(B,\nu)$-bounded expansion of order $r$} if for some
constant $C$ we have
\begin{equation}\label{eq_B_nu_defs_summ}
f(k,n) = c_0(k)+\cdots+c_{r-1}(k)+ O(1) c_r(k)/n^r,
\end{equation} 
whenever $f(k,n)$ is defined and $1\le k\le n^{1/2}/C$, where
for $0\le i\le r-1$, the $c_i(k)$ are $(B,\nu)$-bounded and $c_r(k)$
is of growth $\mu_1(B)$.
Furthermore, such an expansion is called {\em $(B,\nu)$-Ramanujan}
if for $0\le i\le r-1$, the $c_i(k)$ are {\em $(B,\nu)$-Ramanujan}.

Typically our functions $f(k,n)$ as in
\eqref{eq_B_nu_defs_summ} are defined for all $k\in\naturals$
and $n\in N$ for an infinite set $N\subset\naturals$ representing
the possible degrees of our random covering maps in the model
$\{\cC_n(B)\}_{n\in N}$ at hand.

\subsection{Tangles}
\label{su_tangles}

A {\em $(\ge\nu)$-tangle} is any 
connected graph, $\psi$, with $\mu_1(\psi)\ge\nu$,
where $\mu_1(\psi)$ denotes the Perron-Frobenius eigenvalue of $H_B$;
a {\em $(\ge\nu,<r)$-tangle} is any $(\ge\nu)$-tangle of order less than
$r$;
similarly for $(>\nu)$-tangles, i.e.,
$\psi$ satisfying the weak inequality $\mu_1(\psi)>\nu$,
and for $(>\nu,r)$-tangles.
We use ${\rm TangleFree}(\ge\nu,<r)$ to denote those graphs that don't
contain a subgraph that is $(\ge\nu,<r)$-tangle, and
${\rm HasTangles}(\ge\nu,<r)$ for those that do; we
never use $(>\nu)$-tangles in defining TangleFree and HasTangles,
for the technical reason
(see Article~III or Lemma~9.2 of \cite{friedman_alon}) that
for $\nu>1$ and any $r\in\naturals$
that there are only finitely many 
$(\ge\nu,<r)$-tangles, up to isomorphism, that are minimal
with respect to inclusion\footnote{
  By contrast, there are infinitely many minimal $(>\nu,<r)$-tangles
  for some values of $\nu>1$ and $r$: indeed, consider any connected pruned
  graph $\psi$, and set $r=\ord(\psi)+2$, $\nu=\mu_1(\psi)$.  Then if
  we fix two vertices in $\psi$ and let $\psi_s$ be the graph that is
  $\psi$ with an additional edge of length $s$ between these two 
  vertices, then $\psi_s$ is an $(>\nu,<r)$-tangle.  However, if
  $\psi'$ is $\psi$ with any single edge deleted, and $\psi'_s$ is 
  $\psi_s$ with this edge deleted, then one can show that
  $\mu_1(\psi'_s)<\nu$ for $s$ sufficiently large.  It follows that
  for $s$ sufficiently large, $\psi_s$ are minimal $(>\nu,<r)$-tangles.
}.

\subsection{$B$-Graphs, Ordered Graphs, and Strongly Algebraic Models}
\label{su_B_ordered_strongly_alg}

An {\em ordered graph}, $G^\og$, is a graph, $G$, endowed with an
{\em ordering}, meaning
an orientation (i.e., $\iota_G$-orbit representatives), 
$\Eor_G\subset\Edir_G$, 
and total orderings of $V_G$ and $E_G$;
a walk, $w=(v_0,\ldots,e_k,v_k)$ in a graph endows $\ViSu(w)$ with a
{\em first-encountered} ordering:
namely, $v\le v'$ if the first occurrence of $v$ comes before that
of $v'$ in the sequence $v_0,v_1,\ldots,v_k$,
similarly for $e\le e'$, and we orient each edge in the
order in which it is first traversed (some edges may be traversed
in only one direction).
We use $\ViSu^\og(w)$ to refer to $\ViSu(w)$ with this ordering.

A {\em morphism} $G^\og\to H^\og$ of ordered graphs is a morphism
$G\to H$ that respects the ordering in the evident fashion.
We are mostly interested in {\em isomorphisms} of ordered graphs;
we easily see that any isomorphism $G^\og\to G^\og$ must be the
identity morphism; it follows that if $G^\og$ and $H^\og$ are
isomorphic, then there is a unique isomorphism $G^\og\to H^\og$.

If $B$ is a graph, then a $B$-graph, $G_\Bg$, is a graph $G$ endowed 
with a map $G\to B$ (its {\em $B$-graph} structure).
A {\em morphism} $G_\Bg\to H_\Bg$ of $B$-graphs is a morphism
$G\to H$ that respects the $B$-structures in the evident sense.
An {\em ordered $B$-graph}, $G^\og_\Bg$, is a graph endowed with
both an ordering and a $B$-graph structure; a morphism of
ordered $B$-graphs is a morphism of the underlying graphs that
respects both the ordering and $B$-graph structures.
If $w$ is a walk in a $B$-graph, $G_\Bg$, we use $\ViSu_\Bg(w)$ to denote
$\ViSu(w)$ with the $B$-graph structure it inherits from $G$ in
the evident sense; we use $\ViSu_\Bg^\og(w)$ to denote
$\ViSu_\Bg(w)$ with its first-encountered ordering.

At times we drop the superscript $\,^\og$ and the subscript $\,_\Bg$;
for example, we write $G\in\Coord_n(B)$ instead of $G_\Bg\in\cC_n(B)$
(despite the fact that we constantly utilize
the $B$-graph structure on elements of
$\Coord_n(B)$).

A $B$-graph $G_\Bg$ is {\em covering} or {\'etale} if its structure
map $G\to B$ is.

If $\pi\from S\to B$ is a $B$-graph, we use
$\mec a=\mec a_{S_\Bg}$ to denote the vector
$\Edir_B\to\integers_{\ge 0}$ given by
$a_{S_\Bg}(e) = \# \pi^{-1}(e)$;
since $a_{S_\Bg}(\iota_B e) = a_{S_\Bg}(e)$ for all $e\in\Edir_B$,
we sometimes view $\mec a$ as a function $E_B\to\integers_{\ge 0}$, i.e.,
as the function taking $\{e,\iota_B e\}$ to 
$a_{S_\Bg}(e)=a_{S_\Bg}(\iota_B e)$.
We similarly define $\mec b_{S_\Bg}\from V_B\to\integers_{\ge 0}$ by
setting $b_{S_\Bg}(v) = \#\pi^{-1}(v)$.
If $w$ is a walk in a $B$-graph, we set $\mec a_w$ to be
$\mec a_{S_\Bg}$ where $S_\Bg=\ViSu_\Bg(w)$, and similarly for $\mec b_w$.
We refer to $\mec a,\mec b$ (in either context) as
{\em $B$-fibre counting functions}.

If $S_\Bg^\og$ is an ordered $B$-graph and $G_\Bg$ is a $B$-graph, we 
use $[S_\Bg^\og]\cap G_\Bg$ to denote the set of ordered graphs ${G'}_\Bg^\og$
such that $G'_\Bg\subset G_\Bg$ and ${G'}_\Bg^\og\isom S_\Bg^\og$
(as ordered $B$-graphs); this set is naturally identified with the
set of injective morphisms $S_\Bg\to G_\Bg$, and the cardinality of these
sets is independent of the ordering on $S_\Bg^\og$.


A $B$-graph, $S_\Bg$, or an ordered $B$-graph, $S_\Bg^\og$,
{\em occurs in a model $\{\cC_n(B)\}_{n\in N}$}
if for all sufficiently large
$n\in N$, $S_\Bg$ is isomorphic to a $B$-subgraph of some element
of $\cC_n(B)$; similary a graph, $S$, {\em occurs in 
$\{\cC_n(B)\}_{n\in N}$} if it can be endowed with a $B$-graph
structure, $S_\Bg$, that occurs in 
$\{\cC_n(B)\}_{n\in N}$.

A model $\{\cC_n(B)\}_{n\in N}$ of coverings of $B$ is {\em strongly
algebraic} if
\begin{enumerate}
\item for each $r\in\naturals$
there is a function, $g=g(k)$, of growth $\mu_1(B)$
such that if $k\le n/4$ we have
\begin{equation}\label{eq_algebraic_order_bound}
\EE_{G\in\cC_n(B)}[ \snbc_{\ge r}(G,k)] \le
g(k)/n^r
\end{equation}
where $\snbc_{\ge r}(G,k)$ is the number of SNBC walks of length
$k$ in $G$ whose visited subgraph is of order at least $r$;
\item
for any $r$ there exists
a function $g$ of growth $1$ and real $C>0$ such that the following
holds:
for any ordered $B$-graph, $S_\Bg^\og$, that is pruned and of
order less than $r$,
\begin{enumerate}
\item
if $S_\Bg$ occurs in $\cC_n(B)$, then for
$1\le \#\Edir_S\le n^{1/2}/C$,
\begin{equation}\label{eq_expansion_S}
\EE_{G\in\cC_n(B)}\Bigl[ \#\bigl([S_\Bg^\og]\cap G\bigr) \Bigr]
=
c_0 + \cdots + c_{r-1}/n^{r-1}
+ O(1) g(\# E_S) /n^r
\end{equation} 
where the $O(1)$ term is bounded in absolute value by $C$
(and therefore independent of $n$ and $S_\Bg$), and
where $c_i=c_i(S_\Bg)\in\reals$ such that
$c_i$ is $0$ if $i<\ord(S)$ and $c_i>0$ for $i=\ord(S)$;
and
\item
if $S_\Bg$ does not occur in $\cC_n(B)$, then for any
$n$ with $\#\Edir_S\le n^{1/2}/C$,
\begin{equation}\label{eq_zero_S_in_G}
\EE_{G\in\cC_n(B)}\Bigl[ \#\bigl([S_\Bg^\og]\cap G\bigr) \Bigr]
= 0 
\end{equation} 
(or, equivalently, no graph in $\cC_n(B)$ has a $B$-subgraph isomorphic to
$S_\Bg^\og$);
\end{enumerate}
\item
$c_0=c_0(S_\Bg)$ equals $1$ if $S$ is a cycle (i.e., $\ord(S)=0$ and
$S$ is connected) that occurs in $\cC_n(B)$;
\item
$S_\Bg$ occurs in $\cC_n(B)$ iff $S_\Bg$ is an \'etale $B$-graph
and $S$ has no half-loops; and
\item
there exist
polynomials $p_i=p_i(\mec a,\mec b)$ such that $p_0=1$
(i.e., identically 1), and for every
\'etale $B$-graph, $S_\Bg^\og$, we have that
\begin{equation}\label{eq_strongly_algebraic}
c_{\ord(S)+i}(S_\Bg) = p_i(\mec a_{S_\Bg},\mec b_{S_\Bg}) \ .
\end{equation}
\end{enumerate}
Notice that condition~(3), regarding $S$ that are cycles, is implied
by conditions~(4) and~(5); we leave in condition~(3) since this makes the
definition of {\em algebraic} (below) simpler.
Notice that \eqref{eq_expansion_S} and \eqref{eq_strongly_algebraic}
are the main reasons that we work with
ordered $B$-graphs: indeed, the coefficients depend only on
the $B$-fibre counting function $\mec a,\mec b$, which 
depend on the structure of
$S_\Bg^\og$ as a $B$-graph; this is not true if we don't work with
ordered graphs: i.e.,
\eqref{eq_expansion_S} fails to
hold if we replace $[S_\Bg^\og]$
with $[S_\Bg]$ (when $S_\Bg$ has nontrivial automorphisms), where
$[S_\Bg]\cap G$ refers to the number of $B$-subgraphs of $G$ isomorphic
to $S_\Bg$; the reason is that
$$
\#[S_\Bg^\og]\cap G_\Bg = \bigl( \#{\rm Aut}(S_\Bg)\bigr)
\bigl( \#[S_\Bg]\cap G_\Bg \bigr)
$$
where ${\rm Aut}(S_\Bg)$ is the group of automorphisms of $S_\Bg$, 
and it is $[S_\Bg^\og]\cap G_\Bg$ rather than $[S_\Bg]\cap G_\Bg$
that turns out to have the ``better'' properties;
see Section~6 of Article~I for examples.
Ordered graphs are convenient to use for a number of other reasons.

\ignore{
\myDeleteNote{Stuff deleted here and below on September 13, 2018.}
}

\subsection{Homotopy Type}

The homotopy type of a walk and of an ordered subgraph are defined
by {\em suppressing} its ``uninteresting'' vertices of degree two;
examples are given in Section~6 of Article~I.
Here is how we make this precise.

A {\em bead} in a graph is a vertex of degree two that is not
incident upon a self-loop.
Let $S$ be a graph and $V'\subset V_S$ be a {\em proper bead subset} of 
$V_S$,
meaning that $V'$ consists only of beads of $V$,
and that no connected component of $S$ has all its vertices in
$V'$ (this can only happen for connected components of $S$ that
are cycles);
we define the {\em bead suppression} $S/V'$ to be the following
graph: (1) its
vertex set $V_{S/V'}$
is $V''=V_S\setminus V'$, (2) its directed edges, $\Edir_{S/V'}$ consist
of
the {\em $V$'-beaded paths}, i.e., non-backtracking walks
in $S$ between elements of $V''$ whose intermediate vertices lie in $V'$,
(3) $t_{S/V'}$ and $h_{S/V'}$ give the first and last vertex of
the beaded path, and (4) $\iota_{S/V'}$ takes a beaded path
to its reverse walk
(i.e., takes $(v_0,e_1,\ldots,v_k)$ to
$(v_k,\iota_S e_k,\ldots,\iota_S e_1,v_0)$).
One can recover $S$ from the suppression $S/V'$ for pedantic reasons,
since we have defined its directed edges to be beaded paths of $S$.
If $S^\og=\ViSu^\og(w)$ where $w$ is a non-backtracking walk,
then the ordering of $S$ can be inferred by the naturally
corresponding order on $S/V'$, and we use $S^\og/V'$ to denote
$S/V'$ with this ordering.

Let $w$ be a non-backtracking walk in a graph, and 
$S^\og=\ViSu^\og(w)$ its visited
subgraph; the {\em reduction} of $w$ is the ordered graph,
$R^\og$, denoted $S^\og/V'$,
whose underlying graph is
$S/V'$ where $V'$ is the set of beads of $S$ except
the first and last vertices of $w$ (if one or both are beads),
and whose ordering is naturally arises from that on $S^\og$;
the {\em edge lengths} of $w$ is the function $E_{S/V'}\to\naturals$
taking an edge of $S/V'$ to the length of the beaded path it represents
in $S$;
we say that $w$ is {\em of homotopy type} $T^\og$ for any ordered
graph $T^\og$ that is isomorphic to $S^\og/V'$; in this case
the lengths of $S^\og/V'$ naturally give lengths $E_T\to\naturals$
by the unique isomorphism from $T^\og$ to $S^\og/V'$.
If $S^\og$ is the visited subgraph of a non-backtracking walk,
we define the reduction, homotopy type, and edge-lengths of $S^\og$ to
be that of the walk, since these notions depend only on $S^\og$ and
not the particular walk.

If $T$ is a graph and $\mec k\from E_T\to\naturals$ a function, then
we use $\VLG(T,\mec k)$ (for {\em variable-length graph}) to denote
any graph obtained from $T$ by gluing in a path of length $k(e)$
for each $e\in E_T$.  If $S^\og$ is of homotopy type $T^\og$
and $\mec k\from E_T\to \naturals$ its edge lengths,
then $\VLG(T,\mec k)$ is isomorphic to $S$ (as a graph).
Hence the construction of variable-length graphs is a sort of
inverse to bead suppression.

If $T^\og$ is an ordering on $T$ that arises as the first encountered
ordering of a non-backtracking walk on $T$ (whose visited subgraph
is all of $T$), then this ordering gives rise to a natural
ordering on $\VLG(T,\mec k)$ that we denote $\VLG^\og(T^\og,\mec k)$.
Again, this ordering on the variable-length graph is a sort of
inverse to bead suppression on ordered graphs.

\subsection{$B$-graphs and Wordings}

If $w_B=(v_0,e_1,\ldots,e_k,v_k)$ with $k\ge 1$ is a walk in a graph
$B$, then we can identify
$w_B$ with the string $e_1,e_2,\ldots,e_k$ over the alphabet
$\Edir_B$.
For technical reasons, the definitions below of
a {\em $B$-wording} and 
the {\em induced wording}, are given as strings over $\Edir_B$ rather
than the full alternating string of vertices and directed edges.
The reason is that 
doing this gives the correct notion of the {\em eigenvalues} of
an algebraic model (defined below).

Let $w$ be a non-backtracking walk in a $B$-graph, whose reduction
is $S^\og/V'$, and let
$S_\Bg^\og=\ViSu_\Bg^\og$.
Then the {\em wording induced by $w$} on $S^\og/V'$ is
the map $W$ from $\Edir_{S/V'}$ to strings in $\Edir_B$
of positive length, 
taking a
directed edge $e\in\Edir_{S/V'}$ to the string of $\Edir_B$ edges
in the non-backtracking walk in $B$
that lies under the walk in $S$ that it represents.
Abstractly, we say that a {\em $B$-wording} of a graph $T$
is a map $W$ from $\Edir_T$ to words over the alphabet
$\Edir_B$ that represent (the directed edges of)
non-backtracking walks in $B$ such that
(1) $W(\iota_T e)$ is the reverse word (corresponding to
the reverse walk) in $B$ of $W(e)$, 
(2) if $e\in\Edir_T$ is a half-loop, then $W(e)$ is of length one
whose single letter is a half-loop, and
(3) the tail of the first directed edge in $W(e)$ 
(corresponding to the first vertex in the associated walk in $B$)
depends only on $t_T e$;
the {\em edge-lengths} of $W$ is the function $E_T\to\naturals$
taking $e$ to the length of $W(e)$.
[Hence the wording induced by $w$ above is, indeed, a $B$-wording.]

Given a graph, $T$, and a $B$-wording $W$, there is a $B$-graph,
unique up to isomorphism, whose underlying graph is $\VLG(T,\mec k)$
where $\mec k$ is the edge-lengths of $W$, and where the $B$-graph
structure maps the non-backtracking walk in $\VLG(T,\mec k)$
corresponding to an $e\in\Edir_T$ to the non-backtracking walk in $B$
given by $W(e)$.
We denote any such $B$-graph by $\VLG(T,W)$; again this is
a sort of inverse to starting with a non-backtracking walk
and producing the wording it induces on its visited subgraph.

Notice that if $S_\Bg^\og=\VLG(T^\og,W)$ for a $B$-wording, $W$,
then the $B$-fibre counting functions
$\mec a_{S_\Bg}$ and $\mec b_{S_\Bg}$ can be
inferred from $W$, and we may therefore write $\mec a_W$ and
$\mec b_W$.

\subsection{Algebraic Models}

By a $B$-type we mean a pair $T^{\rm type}=(T,\cR)$ consisting
of a graph, $T$, and a map from $\Edir_T$ to the set
of regular languages over the alphabet $\Edir_B$ (in the sense of regular
language theory) such that
(1) all words in $\cR(e)$ are positive length strings corresponding to
non-backtracking walks in $B$, 
(2) if for $e\in\Edir_T$ we have $w=e_1\ldots e_k\in\cR(e)$,
then $w^R\eqdef \iota_B e_k\ldots\iota_B e_1$ lies in $\cR(\iota_T e)$,
and (3) if $W\from \Edir_T\to(\Edir_B)^*$ (where $(\Edir_B)^*$ is
the set of strings over $\Edir_B$) satisfies
$W(e)\in\cR(e)$ and $W(\iota_T e)=W(e)^R$ for all $e\in \Edir_T$,
then $W$ is a $B$-wording.
A $B$-wording $W$ of $T$ is {\em of type $T^{\rm type}$} if
$W(e)\in\cR(e)$ for each $e\in\Edir_T$.

Let $\cC_n(B)$ be a model that satisfies (1)--(3) of the definition
of strongly algebraic.
If $\cT$ a subset of $B$-graphs,
we say that the model is {\em algebraic restricted to $\cT$}
if 
either all $S_\Bg\in\cT$ occur in $\cC_n(B)$ or they all do not,
and if so
there are polynomials $p_0,p_1,\ldots$ such that
$c_i(S_\Bg)=p_i(S_\Bg)$ for any $S_\Bg\in\cT$. 
We say that $\cC_n(B)$ is {\em algebraic} if 
\begin{enumerate}
\item
setting $h(k)$ to be
the number of $B$-graph isomorphism classes of \'etale $B$-graphs
$S_\Bg$ such that $S$ is a cycle of length $k$ and $S$ does
not occur in $\cC_n(B)$, we have that 
$h$ is a function of growth $(d-1)^{1/2}$; and
\item
for any
pruned, ordered graph, $T^\og$, there is a finite number of
$B$-types, $T_j^{\rm type}=(T^\og,\cR_j)$, $j=1,\ldots,s$, 
such that (1) any $B$-wording, $W$, of $T$ belongs to exactly one
$\cR_j$, and
(2) $\cC_n(B)$ is algebraic when restricted to $T_j^{\rm type}$.
\end{enumerate}

[In Article~I we show that
if instead each $B$-wording belong to 
{\em at least one} $B$-type $T_j^{\rm type}$, then one can choose a
another set of
$B$-types that satisfy (2) and where each $B$-wording belongs
to {\em a unique} $B$-type;
however, the uniqueness
is ultimately needed in our proofs,
so we use uniqueness in our definition of algebraic.]

We remark that one can say that a walk, $w$, in a $B$-graph,
or an ordered $B$-graphs, $S_\Bg^\og$, is of {\em homotopy type $T^\og$},
but when $T$ has non-trivial automorphism one {\em cannot} say
that is of $B$-type $(T,\cR)$ unless---for example---one orders
$T$ and speaks of an {\em ordered $B$-type}, $(T^\og,\cR)$.
[This will be of concern only in Article~II.]

We define the {\em eigenvalues} of a regular language, $R$, to be the minimal
set $\mu_1,\ldots,\mu_m$ such that for any $k\ge 1$,
the number of words of length $k$ in the language
is given as
$$
\sum_{i=1}^m p_i(k)\mu_i^k
$$
for some polynomials $p_i=p_i(k)$, with the convention that
if $\mu_i=0$ then $p_i(k)\mu_i^k$ refers to any function that 
vanishes for $k$ sufficiently large (the reason for this is that
a Jordan block of eigenvalue $0$ is a nilpotent matrix).
Similarly, we define the eigenvalues of a $B$-type $T^{\rm type}=(T,\cR)$
as the union of all the eigenvalues of the $\cR(e)$.
Similarly a {\em set of eigenvalues} of a graph, $T$
(respectively, an algebraic model, $\cC_n(B)$)
is
any set containing the eigenvalues containing the eigenvalues
of some choice of $B$-types used in the definition of algebraic
for $T$-wordings (respectively, for $T$-wordings for all $T$).

[In Article~V we prove that all of our basic models are algebraic;
some of our basic models, such as the
permutation-involution model and the cyclic models, are not
strongly algebraic.]

We remark that a homotopy type, $T^\og$,
of a non-backtracking walk, can only have beads as its first or last 
vertices; however, in the definition of algebraic we require
a condition on {\em all pruned graphs}, $T$, 
which includes $T$ that may have many beads and may not be connected;
this is needed
when we define homotopy types of pairs in Article~II.

\subsection{SNBC Counting Functions}

If $T^\og$ is an ordered graph and $\mec k\from E_T\to\naturals$, 
we use $\SNBC(T^\og,\mec k;G,k)$ to denote the set of SNBC walks in $G$
of length $k$ and of homotopy type $T^\og$ and edge lengths $\mec k$.
We similarly define
$$
\SNBC(T^\og,\ge\bec\xi;G,k) \eqdef 
\bigcup_{\mec k\ge\bec\xi} \SNBC(T^\og,\mec k;G,k)
$$
where $\mec k\ge\bec\xi$ means that $k(e)\ge\xi(e)$ for all $e\in E_T$.
We denote the cardinality of these sets by replacing $\SNBC$ with
$\snbc$;
we call $\snbc(T^\og,\ge\bec\xi;G,k)$ the set of 
{\em $\bec\xi$-certified
traces of homotopy type $T^\og$ of length $k$ in $G$};
in Article~III we will refer to certain $\bec\xi$ as {\em certificates}.

\section{Main Results}
\label{se_main_results}

In this section we give some more definitions and explain the
results we prove in this article and the next article in this series.
We also state the main results
from Articles~III and IV that we will need to quote here;
more details about these results can be found in Articles~III and IV, and
some rough remarks on these results and articles can be found in Article~I.

\subsection{The First Main Theorem}

If $B$ is a graph, $\|A_{\widehat B}\|_2$ denotes
the $L^2$ norm of the adjacency operator
on a universal cover, $\widehat B$, of $B$; it is well-known that
if $B$ is $d$-regular,
then $\|A_{\widehat B}\|_2=2\sqrt{d-1}$
(see, for example, \cite{mohar_woess}).
If $\pi\from G\to B$ is a covering map graphs,
and $\epsilon>0$, the
{\em $\epsilon$-non-Alon multiplicity of $G$
relative to $B$} is
$$
{\rm NonAlon}_B(G;\epsilon) \eqdef
\# \bigl\{\lambda\in\specnew_B(A_G)\ \bigm|\
|\lambda|>
\|A_{\widehat B}\|_2 +\epsilon\bigr\} ,
$$
where
the above $\lambda$ are counted with their multiplicity in
$\specnew_B(A_G)$.

Here is our first main theorem.

\begin{theorem}\label{th_first_main_thm}
Let $B$ be a $d$-regular graph, and $\{\cC_n(B)\}_{n\in N}$ an algebraic
model over $B$.
Then for any
$\epsilon>0$ there is a constant $C=C(\epsilon)$ for which
$$
\Prob_{G\in\cC_n(B)}[ {\rm NonAlon}_B(G;\epsilon)>0 ]
\le  C(\epsilon)/n \ .
$$
\end{theorem}
In fact, we conjecture that for the above probability
there are matching upper and lower bounds,
within a constant depending on $\epsilon$ (but not on $n$),
that we now describe.  It will be convenient to first recall the
main results from Article~III.

\subsection{Results Needed from Article~III}

Let us recall the main theorem of Article~III.

%
\begin{theorem}\label{th_main_tech_result}
Let $B$ be a connected graph with
$\mu_1(B)>1$, and let
$\{\cC_n(B)\}_{n\in N}$ be
an algebraic model over $B$.
Let $r>0$ be an integer and $\nu\ge\mu_1^{1/2}(B)$ be a real number.
Then
\begin{equation}\label{eq_main_tech_result1}
f(k,n)\eqdef
\EE_{G\in\cC_n(B)}[ \II_{{\rm TangleFree}(\ge\nu,<r)}(G) \Trace(H^k_G) ]
\end{equation}
has a $(B,\nu)$-bounded expansion to order $r$,
$$
f(k,n)=c_0(k)+\cdots+c_{r-1}(k)/n^{r-1}+O(1)c_r(k)/n^r,
$$
where
\begin{equation}\label{eq_c_0}
c_0(k)=\sum_{k'|k} \Trace(H_B^{k'}) - h(k)
\end{equation} 
where the sum is over all positive integers, $k'$, dividing $k$
and where $h(k)$ is of growth $(d-1)^{1/2}$; hence
\begin{equation}\label{eq_c_zero_roughly}
c_0(k) = \Trace(H_B^k) + \tilde h(k)
\end{equation} 
where $\tilde h(k)$ is a function of growth $(d-1)^{1/2}$;
furthermore, the larger
bases of each $c_i(k)$ (with respect to $\mu_1^{1/2}(B)$)
is some subset of the eigenvalues
of the model.
Also, the function $h(k)$ in \eqref{eq_c_0} is precisely
the function described in condition~(1) of the definition
of {\em algebraic model}.
Finally, for any $r'\in\naturals$ the function
\begin{equation}\label{eq_main_tech_result2}
\widetilde f(n)  \eqdef
\EE_{G\in\cC_n(B)}[ \II_{{\rm TangleFree}(\ge\nu,<r')}(G)]
=
\Prob_{G\in\cC_n(B)}[ G\in {\rm TangleFree}(\ge\nu,<r') ]
\end{equation}
has an asymptotic expansion in $1/n$ to any order $r$,
$$
\widetilde c_0+\cdots+\widetilde c_{r-1}/n^{r-1}+O(1)/n^r ;
$$
where $\widetilde c_0=1$; furthermore, if $j_0$ is the
smallest order of a $(\ge\nu)$-tangle occurring in $\cC_n(B)$,
then
$\widetilde c_j=0$ for $1\le j<j_0$ and
$\widetilde c_j>0$ for $j=j_0$ (provided that $r\ge j_0+1$ so that
$\widetilde c_{j_0}$ is defined).
\end{theorem}

We will also need the following result of Article~III, whose
proof is related to the result \eqref{eq_main_tech_result2}
(actually both results are special cases of a more general result
proven there).

\begin{theorem}\label{th_extra_needed}
Let $\cC_n(B)$ be an algebraic model over a graph, $B$, and
let $S_\Bg$ be a connected, pruned graph of positive order
that occurs in this model (recall that this
means that for some $n$ and some $G\in\cC_n(B)$, $G_\Bg$ has a subgraph
isomorphic to $S_\Bg$).  Then for some constant, $C'$, and
$n$ sufficiently large,
$$
\Prob_{G\in\cC_n(B)}\Bigl[ [S_\Bg]\cap G\ne\emptyset
\Bigr] \ge
C' n^{-\ord(S_\Bg)}.
$$
\end{theorem}
Of course, by definition of an algebraic model we know that
for any ordering $S^\og$ on $S$ we have
$$
\EE_{G\in\cC_n(B)}\Bigl[ \#[S_\Bg^\og]\cap G
\Bigr] =
n^{-\ord(S)} \bigl(c + o(1/n) \bigr)
$$
for some $c=c(S_\Bg)>0$ (and actually $c=1$ in all of our basic models);
and from this it follows that (see Article~I or III)
$$
\EE_{G\in\cC_n(B)}\Bigl[ \#[S_\Bg]\cap G
\Bigr] =
\EE_{G\in\cC_n(B)}\Bigl[ \#[S_\Bg^\og]\cap G
\Bigr] / \bigl( \# {\rm Aut}(S_\Bg) \bigr)
$$ 
is also proportional to $n^{-\ord(S)}$.
The idea behind the proof of Theorem~\ref{th_extra_needed} 
in Article~III is that by
inclusion-exclusion one can show that
the probability that a $G\in\cC_n(B)$ contains
two or more subgraphs isomorphic to $S_\Bg$ is $O(n^{-1-\ord(S)})$.
However, Article~III develops more powerful inclusion-exclusion
tools of this sort and proves
theorems that contain Theorem~\ref{th_extra_needed}
as a special case.

\subsection{The Tangle Power of a Model}

%
\begin{definition}\label{de_tau_tang}
Let $\{\cC_n(B)\}_{n\in N}$ be a model over a graph, $B$ with $\mu_1(B)>1$.
By the {\em tangle power of $\{\cC_n(B)\}$}, denoted $\tau_{\rm tang}$,
we mean the smallest order, $\ord(S)$, of any graph, $S$, that
occurs in $\{\cC_n(B)\}$ and satisfies
$\mu_1(S)>\mu_1^{1/2}(B)$.
\end{definition}
The tangle power is finite when $\mu_1(B)>1$, because if $G\in\cC_n(B)$
for some $n\in N$, then $G$ occurs in $\cC_n(B)$ 
and
$\mu_1(G)=\mu_1(B)>\mu_1^{1/2}(B)$; 
hence if $\mu_1(B)>1$, the tangle power of
$\cC_n(B)$ is at most the minimum order of such $G$.
The restriction that $\mu_1(B)>1$ is not a serious restriction, because
we are only interested in $B$ connected and pruned, and hence 
$\mu_1(B)>1$ unless $B$ is a cycle, which is not of interest to us.

The tangle power is relatively easy to bound from below.  In fact,
in Article~VI we use
the results of Section~6.3 of \cite{friedman_alon}
to prove that for any algebraic model over a $d$-regular graph, $B$,
$$
\tau_{\rm tang} \ge  m=m(d)
$$
where
$$
m(d) =
\Bigl\lfloor \bigl( (d-1)^{1/2} - 1 \bigr)/2  \Bigr\rfloor +1 ,
$$
and equality holds 
for each even $d\ge 4$ in the case where $B$ is a bouquet of $d/2$
whole loops.

\ignore{\tiny\red
The results on the $\widetilde c_i$ in the asymptotic expression for 
\eqref{eq_main_tech_result2}
in Theorem~\ref{th_main_tech_result} implies that
that for every $\nu \ge  \mu_1(B)$ and
$r\in\naturals$ we have
$$
C'(\nu,r) n^{-j}
\le \Prob_{G\in\cC_n(B)}[ 
{\rm HasTangles}(\ge\nu,<r)
 ] \le
C(\nu,r) n^{-j}
$$
where $j$ is the smallest order of any $(\ge\nu,<r)$ tangle;
if we restrict $\nu$ to satisfy the strict inequality
$\nu>\mu_1(B)$, then $j\ge\tau_{\rm tang}$, and therefore
\begin{equation}\label{eq_Has_Tangles_non_Alon_upper}
C'(\nu,r) n^{-\tau_{\rm tang}}  \le
\Prob_{G\in\cC_n(B)}[ 
{\rm HasTangles}(\ge\nu,<r)
 ] \le
C(\nu,r) n^{-\tau_{\rm tang}} 
\end{equation} 
Notice that the second inequality implies that 
\begin{equation}\label{eq_trivial_prob_combo}
\Prob_{G\in\cC_n(B)}\Bigl[ 
\bigl( G \in {\rm HasTangles}(\ge\nu,<r) \bigr)
\ \mbox{\rm and}\  %
\bigl({\rm NonAlon}_B(G;\epsilon')>0  \bigr)
\Bigr] \ge
C(\nu,r) n^{-\tau_{\rm tang}}  ,
\end{equation} 
since the probability in \eqref{eq_trivial_prob_combo}
is clearly bounded from above by that in
\eqref{eq_Has_Tangles_non_Alon_upper}.
In this article we will
prove the following a matching lower bound.
}

The following theorem explains our interest in $\tau_{\rm tang}$
regarding the relativized Alon conjecture.

\begin{theorem}\label{th_hastangles_lower_bound}
Let $\{\cC_n(B)\}_{n\in N}$ be an algebraic model of tangle
power $\tau_{\rm tang}$ over a $d$-regular graph, $B$.
Let $S$ be a connected graph that occurs in $\cC_n(B)$ with
$\ord(S)=\tau_{\rm tang}$ and $\mu_1(S)>(d-1)^{1/2}$, and set
\begin{equation}\label{eq_epsilon_0_S} 
\epsilon_0 = \mu_1(S)+\frac{d-1}{\mu_1(S)} - 2(d-1)^{1/2}.
\end{equation} 
Then there is a constant $C'$ and $n_0$ such that
for any $r\in\naturals$ and real $\nu$ with
\begin{equation}\label{eq_r_nu_for_hastangles_lower_bound}
r\ge \ord(S), \quad (d-1)^{1/2}< \nu \le \mu_1(S),
\end{equation} 
for any $n\ge n_0$ we have
\begin{equation}\label{eq_Has_Tangles_non_Alon_lower}
\Prob_{G\in\cC_n(B)}\Bigl[ 
\bigl( G \in {\rm HasTangles}(\ge\nu,<r) \bigr)
\ \mbox{\rm and}\  %
\bigl({\rm NonAlon}_B(G;\epsilon_0/2)>0  \bigr)
\Bigr] \ge
C' n^{-\tau_{\rm tang}}  .
\end{equation} 
Furthermore, for any $r,\nu$ satisfying 
\eqref{eq_r_nu_for_hastangles_lower_bound}
there is a constant $C=C(\nu,r)$ such that
\begin{equation}\label{eq_Has_Tangles_non_Alon_upper}
\Prob_{G\in\cC_n(B)}\Bigl[ 
\bigl( G \in {\rm HasTangles}(\ge\nu,<r) \bigr)
\ \mbox{\rm and}\  %
\bigl({\rm NonAlon}_B(G;\epsilon)>0  \bigr)
\Bigr] \le
C(\nu,r) n^{-\tau_{\rm tang}}  .
\end{equation} 
\end{theorem}
When we prove this theorem in Section~\ref{se_remaining_proofs};
there we will see that \eqref{eq_Has_Tangles_non_Alon_upper}
is an immediate consequence of Theorem~\ref{th_main_tech_result},
but our proof of
\eqref{eq_Has_Tangles_non_Alon_lower} requires some work, which
generalizes some of the results in
Friedman-Tillich \cite{friedman_tillich_generalized}.

Theorem~\ref{th_hastangles_lower_bound}
implies that for $\epsilon>0$ sufficiently small (namely
$\epsilon\le\epsilon_0/2$ with $\epsilon_0$ as in
\eqref{eq_epsilon_0_S}) 
we have
$$
\Prob_{G\in\cC_n(B)}\Bigl[ 
{\rm NonAlon}_B(G;\epsilon)>0  
\Bigr] \ge
C' n^{-\tau_{\rm tang}}  .
$$
We conjecture that this lower bound has a matching upper bound
to within a constant
depending on $\epsilon>0$, and in Article~VI we will prove this
for our basic models when $B$ is $d$-regular and Ramanujan.
Let us explain this in more detail.

\subsection{The Algebraic Power of a Model}

In this article we combine the results of Articles~I--IV to prove
the following main result.

\begin{theorem}\label{th_rel_Alon_regular2}
Let $\cC_n(B)$ be an
algebraic model
over a $d$-regular graph $B$.
For any $\nu$ with
$(d-1)^{1/2}<\nu<d-1$, let $\epsilon'>0$ be given by
$$
2(d-1)^{1/2}+ \epsilon' = \nu + \frac{d-1}{\nu}.
$$
Then
\begin{enumerate}
\item
there is an
integer $\tau=\tau_{\rm alg}(\nu,r)\ge 1$ such that
for any sufficiently small $\epsilon>0$ there are constants
$C=C(\epsilon),C'>0$ such that for sufficiently large $n$ we have
\begin{equation}\label{eq_rel_alon_expect_lower_and_upper}
n^{-\tau} C'
\le
\EE_{G\in\cC_n(B)}[ \II_{{\rm TangleFree}(\ge\nu,<r)}(G)
{\rm NonAlon}_d(G;\epsilon'+\epsilon) ]
\le
n^{-\tau}   C(\epsilon) ,
\end{equation}
or
\item
for all $j\in\naturals$ and $\epsilon>0$ we have
\begin{equation} 
\label{eq_rel_alon_expect_upper_infinite}
\EE_{G\in\cC_n(B)}[ \II_{{\rm TangleFree}(\ge\nu,<r)}(G)
{\rm NonAlon}_d(G;\epsilon'+\epsilon) ] \le O(n^{-j})
\end{equation} 
in which case we use the notation $\tau_{\rm alg}(\nu,r)=+\infty$.
\end{enumerate}
Moreover, if $\tau=\tau_{\rm alg}(\nu,r)$ is finite, then for
some eigenvalue, $\ell\in\reals$, of the model with $|\ell|>\nu$,
there is a real $C_\ell>0$
such that for sufficiently small $\theta>0$
\begin{equation}\label{eq_new_eigenvalues_near_ell}
\lim_{n\to\infty} 
\EE_{G\in\cC_n(B)}\bigl[\#
\bigl(\specnew_B(H_G)\cap B_{n^{-\theta}}(\ell) \bigr) 
\II_{{\rm TangleFree}(\ge\nu,<r)}(G)
\bigr]
= C_\ell n^{-\tau} + o(n^{-\tau}) .
\end{equation} 
\end{theorem}

[Note that $C'$ in \eqref{eq_rel_alon_expect_lower_and_upper}
is independent of small $\epsilon>0$ since as $\epsilon$
decreases
${\rm NonAlon}_d(G;\epsilon)$ is non-decreasing.]

Notice if $\nu_1\le\nu_2$ and $r_1\ge r_2$ then
$$
\II_{{\rm TangleFree}(\ge\nu_2,<r_2)}(G)  
\le \II_{{\rm TangleFree}(\ge\nu_1,<r_1)}(G) ,
$$
for the simple reason that 
$\II_{{\rm TangleFree}(\ge\nu_2,<r_2)}(G)=1$ implies that
$G$ has no $(\ge\nu_2,<r_2)$-tangles, and hence no
$(\ge\nu_1,<r_1)$-tangles;
then \eqref{eq_rel_alon_expect_lower_and_upper} and
\eqref{eq_rel_alon_expect_upper_infinite} imply that
\begin{equation}\label{eq_tau_alg_nu_r_compare}
\tau_{\rm alg}(\nu_1,r_1) \le
\tau_{\rm alg}(\nu_2,r_2).
\end{equation} 

\begin{definition}\label{de_algebraic_power}
Let $\{\cC_n(B)\}_{n\in N}$ be an algebraic model over a $d$-regular
graph $B$.
For each $r\in\naturals$ and $\nu$ with $(d-1)^{1/2}<\nu<d-1$, let
$\tau_{\rm alg}(\nu,r)$ be as in Theorem~\ref{th_rel_Alon_regular2}.
We define the {\em algebraic power} of the model $\cC_n(B)$ to be
$$
\tau_{\rm alg} =
\max_{\nu>(d-1)^{1/2},r} \tau_{\rm alg}(\nu,r) =
\limsup_{r\to\infty,\ \nu\to(d-1)^{1/2}}
\tau_{\rm alg}(\nu,r)
$$
where $\nu$ tends to $(d-1)^{1/2}$ from above
(and we allow $\tau_{\rm alg}=+\infty$ when this maximum is
unbounded or if $\tau_{\rm alg}(\nu,r)=\infty$ for some $r$ 
and $\nu>(d-1)^{1/2}$).
\end{definition}
Of course, according to Theorem~\ref{th_rel_Alon_regular2},
$\tau_{\rm alg}(\nu,r)\ge 1$ for all $r$ and all relevant $\nu$, and hence
$\tau_{\rm alg}\ge 1$.
Furthermore, since $\tau_{\rm alg}(\nu,r)$ is an integer or $+\infty$, if
$\tau_{\rm alg}$ is finite then for some $\nu_0>(d-1)^{1/2}$ 
and $r_0\in\naturals$ we have
$\tau_{\rm alg}(\nu_0,r_0)=\tau_{\rm alg}$; in this case 
\eqref{eq_tau_alg_nu_r_compare} implies that
for $(d-1)^{1/2}<\nu\le\nu_0$ and $r\ge r_0$ we have
$\tau_{\rm alg}(\nu,r)=\tau_{\rm alg}$.

Since the number of new eigenvalues of a covering map $G\to B$ of degree
$n$ is
$(\#\Edir_G)-(\#\Edir_B)=(n-1)(\#\Edir_B)$, we have that
$$
\Prob_{G\in\cC_n(B)}\Bigl[ 
\bigl( G \in {\rm TangleFree}(\ge\nu,<r) \bigr)
\ \mbox{\rm and}\  %
\bigl({\rm NonAlon}_B(G;\epsilon)>0  \bigr)
\Bigr] 
$$
is between $1$ and $1/((n-1)(\#\Edir_B)$ times
$$
\EE_{G\in\cC_n(B)}[ \II_{{\rm TangleFree}(\ge\nu,<r)}(G)
{\rm NonAlon}_d(G;\epsilon) ] .
$$

\subsection{A More Precise Form of the First Main Theorem}

Taking $\nu\to (d-1)^{1/2}$ (from above) and
$r\to\infty$ 
in~\eqref{eq_Has_Tangles_non_Alon_upper} 
and~\eqref{eq_Has_Tangles_non_Alon_lower},
it is not hard to deduce our second main result,
which refines our first.

%
\begin{theorem}\label{th_rel_Alon_regular}
Let $B$ be a $d$-regular graph, and let
$\cC_n(B)$ be an
algebraic model of tangle power $\tau_{\rm tang}$
and algebraic power $\tau_{\rm alg}$ (both of which are at least $1$).
Let
$$
\tau_1 = \min(\tau_{\rm tang},\tau_{\rm alg}), \quad
\tau_2 = \min(\tau_{\rm tang},\tau_{\rm alg}+1).
$$
Then $\tau_2\ge \tau_1\ge 1$, and
for $\epsilon>0$ sufficiently small there are
$C,C'$ such that for sufficiently large $n$ we have
\begin{equation}\label{eq_first_main}
C' n^{-\tau_2}
\le
\Prob_{G\in\cC_n(B)}\bigl[
{\rm NonAlon}_d(G;\epsilon)>0
\bigr]
\le
C n^{-\tau_1}.
\end{equation}
\end{theorem}

Since ${\rm NonAlon}_d(G;\epsilon)$ is non-increasing in $\epsilon$,
the value $C'$ in \eqref{eq_first_main} is independent of 
sufficiently small $\epsilon$; however, $C$ depends on $\epsilon$.

\subsection{The Main Theorem of Article~VI}

The following theorem will be proven in Article~VI.

\begin{definition}\label{de_Ramanujan}
We say that a $d$-regular graph $B$ is {\em Ramanujan} if all eigenvalues
of $A_B$ lie in
$$
\{d,-d\} \cup \Bigl[ -2\sqrt{d-1}, 2\sqrt{d-1} \Bigr] .
$$
\end{definition}

\begin{theorem}\label{th_second_main_theorem}
Let $\{\cC_n(B)\}_{n\in N}$ be one of our basic models
over $d$-regular Ramanujan
graph, $B$.
Then $\tau_{\rm alg}=+\infty$.
\end{theorem}
The above theorem holds for any $d$-regular Ramanujan graph, $B$, and
for any algebraic model over $B$
that satisfies
a certain weak {\em magnification} condition;
in Article~VI we describe this magnification condition and prove that it holds
for all of our basic models (for any $B$, regular or not).
The proof uses standard counting arguments; for large values of
$d$ the argument is very easy; for
small values of $d$ our argument is a more delicate
calculation similar to those in Chapter~12 of
\cite{friedman_alon}.

We point out that in \cite{friedman_alon}, the upper and lower bounds 
on the probability of ${\rm NonAlon}_d(G;\epsilon)>0$ 
differed by a factor proportional to $n$, rather than a constant,
for random $d$-regular graphs for certain values of $d$, namely
for $d=1+m^2$ for an odd
integer $m\ge 3$, such as $d=10,26,50$.
Hence Theorems~\ref{th_rel_Alon_regular}
and \ref{th_second_main_theorem} improve this factor of $n$
to a constant depending on $\epsilon$ (for such $d$).

\subsection{Conjectures Regarding Theorem~\ref{th_second_main_theorem}}

We make the following successively strong conjectures regarding
Theorem~\ref{th_rel_Alon_regular}.

\begin{conjecture}
Let $\{\cC_n(B)\}_{n\in N}$ be one of our basic models
over $d$-regular graph, $B$.
Then 
\begin{enumerate}
\item
For $\epsilon>0$ sufficiently small there are
$C,C'$ such that for sufficiently large $n$ we have
$$
C' n^{-\tau_{\rm tang}}
\le
\Prob_{G\in\cC_n(B)}\bigl[
{\rm NonAlon}_d(G;\epsilon)>0
\bigr]
\le
C n^{-\tau_{\rm tang}}.
$$
\item
$\tau_{\rm tang} \le \tau_{\rm alg}-1$.
\item
$\tau_{\rm alg}=+\infty$.
\end{enumerate}
\end{conjecture}

Theorem~\ref{th_second_main_theorem} proves the strongest conjecture
in the case where the base graph is Ramanujan.

\subsection{Results Needed from Article~IV}

We recall the main result from Article~IV; we refer to this article
and Article~I for intuition regarding this result.
This result is purely a lemma in probability theory.

\begin{definition}\label{de_matrix_model}
Let $\Lambda_0<\Lambda_1$ be positive real numbers.  
By a {\em $(\Lambda_0,\Lambda_1)$
matrix model} we mean a collection of finite probability spaces
$\{\cM_n\}_{n\in N}$ where $N\subset\naturals$ is an
infinite subset, and where the atoms of $\cM_n$ are
$n\times n$ real-valued matrices whose eigenvalues lie in the set
$$
B_{\Lambda_0}(0) \cup [-\Lambda_1,\Lambda_1]
$$
in $\complex$.
Let $r\ge 0$ be an integer and $K\from\naturals\to\naturals$
be a function such that $K(n)/\log n\to \infty$ as $n\to\infty$.
We say that this model has an
{\em order $r$ expansion} with {\em range $K(n)$} 
(with $\Lambda_0,\Lambda_1$ understood) if
as $n\to\infty$ we have that
\begin{equation}\label{eq_matrix_model_exp}
\EE_{M\in\cM_n}[\Trace(M^k)] =
c_0(k) + c_1(k)/n +\cdots+c_{r-1}(k)/n^{r-1}+ O(c_r(k))/n^r
\end{equation}
for all $k\in\naturals$ with $k\le K(n)$,
where (1) $c_r=c_r(k)$ is of growth $\Lambda_1$,
(2) the constant in the $O(c_r(k))$ is
independent of $k$ and $n$, and
(3) for $0\le i<r$, $c_i=c_i(k)$ is an approximate polyexponential with
$\Lambda_0$ error term and whose larger bases 
(i.e., larger than $\Lambda_0$ in absolute value)
lie in $[-\Lambda_1,\Lambda_1]$;
at times we speak of an {\em order $r$ expansion} without 
explicitly specifying $K$.
When the model has such an expansion,
then we use the notation $L_r$ to refer to the union of all larger
bases of $c_i(k)$ (with respect to $\Lambda_0$) over all $i$ between
$0$ and $r-1$, and call $L_r$ the {\em larger bases (of the order $r$
expansion)}.
\end{definition}
Note that in the above definition,
the larger bases of the $c_i$ are arbitrary, provided that
they lie in $[-\Lambda_1,\Lambda_1]$ (e.g., there is no bound on the
number of bases).
We also note that \eqref{eq_matrix_model_exp} implies that
for fixed $k\in\naturals$,
\begin{equation}\label{eq_c_i_limit_formula}
c_i(k) = \lim_{n\in N,\ n\to\infty}
\Bigl( \EE_{M\in\cM_n}[\Trace(M^k)] -
\bigl( c_0(k) + \cdots+c_{i-1}(k)/n^{i-1} \bigr) \Bigr) n^i
\end{equation} 
for all $i\le r-1$; we conclude that the $c_i(k)$ are uniquely determined,
and that $c_i(k)$ is independent of $r$ for any $r>i$ for which
\eqref{eq_matrix_model_exp} holds.


\begin{theorem}\label{th_sidestep}
Let $\{\cM_n\}_{n\in N}$ be a $(\Lambda_0,\Lambda_1)$-bounded matrix model,
for some real $\Lambda_0<\Lambda_1$, that
for all $r\in\naturals$ has an order $r$ expansion;
let $p_i(k)$ denote the polyexponential part of $c_i(k)$
(with respect to $\Lambda_0$) in \eqref{eq_matrix_model_exp}
(which is independent of $r\ge i+1$ by
\eqref{eq_c_i_limit_formula}).
If $p_i(k)=0$ for all $i\in\integers_{\ge 0}$, then
for all $\epsilon>0$ and $j\in\integers_{\ge 0}$
\begin{equation}\label{eq_largest_j}
\EE{\rm out}_{\cM_n}
\bigl[ B_{\Lambda_0+\epsilon}(0) \bigr]
= O(n^{-j}) .
\end{equation}
Otherwise let $j$ be the smallest integer for which
$p_j(k)\ne 0$.
Then for all $\epsilon>0$, and for all $\theta>0$ sufficiently small we have
\begin{equation}\label{eq_mostly_near_Lambda_0_or_Ls}
\EE{\rm out}_{\cM_n}
\bigl[ B_{\Lambda_0+\epsilon}(0)\cup B_{n^{-\theta}}(L_{j+1}) \bigr]
= o(n^{-j});
\end{equation} 
moreover, if $L=L_{j+1}$ is the (necessarily nonempty) set of bases of $p_j$,
then for each $\ell\in L$ there is a real $C_\ell>0$ such that
\begin{equation}\label{eq_thm_p_j_pure_exp}
p_j(k)=\sum_{\ell\in L} \ell^k C_\ell ,
\end{equation} 
and
for all $\ell\in L$
for sufficiently small $\theta>0$,
\begin{equation}\label{eq_thm_C_ell_as_limit}
\EE{\rm in}_{\cM_n}\bigl[ B_{n^{-\theta}}(\ell) \bigr]  
= n^{-j} C_\ell + o(n^{-j}) .
\end{equation} 
\end{theorem}


\section{The Ihara's Determinantal Formula for Graphs with Half-Loops}
\label{se_ihara}

In this section we prove a generalization of what is often called 
Ihara's Determinantal Formula; our proof follows that of Bass
(see \cite{terras_zeta}, specifically Proposition~19.9, page~172,
and the references on page~43, the beginning of Part~II,
to the work of Ihara, Serre, Sunada, Hashimoto,
and Bass); our generalization allows for graphs to
have half-loops.

Recall that for a graph, $G$, we use
$A_G,H_G$ to respectively denote the adjacency matrix and Hashimoto
(or non-backtracking walk) matrix of $G$.

\begin{theorem}\label{th_Ihara}
For a graph, $G$, and an indeterminate $\mu$ we have
\begin{equation}\label{eq_Ihara_det}
\det(\mu I -H_G) = \det\bigl(\mu^2 I - \mu A_G + (D_G-I) )
(\mu+1)^{o_1(G)}
(\mu^2-1)^{o_2(G)-n} ,
\end{equation}
where $o_1(G)$ is the number of half-loops in $G$ and $o_2(G)$ is the
number of whole-loops and edges that are not whole-loops
(and where the $I$ on the left-hand-side of \eqref{eq_Ihara_det} is the
square identity matrix indexed on $\Edir_G$, and the two $I$'s on
the right-hand-side are the same indexed on $V_G$).
\end{theorem} 
The origin of this formula is 
\cite{ihara}, where it is shown that
$$
\zeta_G(u) = \frac{1}{\det\bigl( I - u A_G + u^2 (D_G-I) \bigr)
(1-u^2)^{-\chi(G)}}
$$
for the Zeta function, $\zeta_G(u)$, of certain graphs, $G$, of interest in
\cite{ihara}; however, the interpretation of this formula in terms
of graph theory occurs
only later (see \cite{serre_arbres}, page~5, or \cite{serre_trees}, page~IX).
The equality
$$
\zeta_G(u)= \frac{1}{\det(I-u H_G)}
$$
(which is relatively easy to see), and the connection to graph theory, 
was made explicit by
\cite{sunada1}, for regular graphs, and \cite{hashimoto1} for
all graphs.
Our proof is a simple adaptation of
Bass's elegant proof this theorem for graphs without half-loops 
(see \cite{bass_elegant,terras_zeta}).
\begin{proof}
Let $u$ be a single indeterminate.
We set $d_h$ the $V_G\times \Edir_G$ matrix whose $(v,e)$ entry is
$1$ if $he=v$, and $0$ otherwise; we similarly define $d_t$.
Introducing an indeterminate $u$,
we easily verify the block matrix equality
$$
\begin{bmatrix} I_{V_G} & 0 \\ d_h^{\rm T} & I_{\Edir_G} \end{bmatrix}
\begin{bmatrix} I_{V_G}(1-u^2) & 0 \\ d_t & I_{\Edir_G}-H_G u \end{bmatrix}
$$
$$
=
\begin{bmatrix} I_{V_G}-A_G u+(D_G-I_{V_G}) u^2 & d_t \\ 
0 & I_{\Edir_G}+ \iota_G u  \end{bmatrix}
\begin{bmatrix} I_{V_G} & 0 \\ 
d_h^{\rm T} - d_t^{\rm T}  u & I_{\Edir_G} \end{bmatrix}
$$
We take determinants of the above, and make use of the identity
$$
\det \begin{bmatrix} M_1 & 0 \\ N_2 & M_2 \end{bmatrix}
=\det \begin{bmatrix} M_1 & N_1 \\ 0 & M_2 \end{bmatrix}
=\det(M_1)\det(M_2)
$$
(for square block matrices $M_1,M_2$ and $N_1,N_2$ of appropriate
size) to conclude that
\begin{equation}\label{eq_ihara_basic}
(1-u^2)^{\#V_G} \det(I_{\Edir_G}-H_G u) =
\det\bigl(I_{V_G}-A_G u + (D_G-I_{V_G}) u^2 \bigr)
\det(I_{\Edir_G} + \iota_G u ).
\end{equation} 
But if $G$ has $o_1$ half-loops and $o_2$ edges (i.e., $\iota_G$ orbits)
that are not half-loops, we have
\begin{equation}\label{eq_ihara_easy}
\det(I_{\Edir_G} + \iota_G u ) = (1-u^2)^{o_2} (1+u)^{o_1} .
\end{equation} 
Combining \eqref{eq_ihara_basic} and \eqref{eq_ihara_easy}, and
substituting $\mu=1/u$ and multiplying by $\mu^{nd}$ yields
\eqref{eq_Ihara_det}.
\end{proof}
The reason we write our proof with $u$ instead of $\mu=1/u$ is 
that this is the usual way the proof is written,
because one usually writes $\zeta_G(u)$,
the Ihara Zeta function \cite{ihara} of a graph, $G$,
as
$$
\zeta_G(u) = \prod_{\frakp} \bigl( 1 - u^{{\rm length}(\frakp)} \bigr)^{-1}
$$
where the product is over all ``primes'' $\frakp$ (primitive, oriented SNBC
walks in $G$)
whereupon it is not hard to see (taking logarithms and considering the
relationship between the trace of $H_G^k$ and primes of length dividing
$k$) that
$$
\zeta_G(u)=\frac{1}{\det(I-u H_G)}.
$$

\section{Our Basic Models are Algebraic}
\label{se_ours_algebraic}

In this section we prove that our basic models are
algebraic.
For ease of reading, we recall the definition of what we call
our basic models.

\subsection{Review of Our Basic Models}

\begin{definition}\label{de_model_conventions}
Let $B$ be a graph.
A {\em model over $B$}
is a family of probability spaces $\{\cC_n(B)\}_{n\in N}$ 
indexed by a parameter
$n$ that ranges over some infinite subset $N\subset\naturals$, such that
the atoms of each $\cC_n(B)$ lie in ${\rm Coord}_n(B)$;
we say that the model is {\em edge-independent} if for any orientation,
$\Eor_B$, of $B$, and each $n\in N$,
if $\{\sigma(e)\}_{e\in\Edir_B}$
are the $\Edir_B\to\cS_n$ maps associated
to the $G\in\cC_n(B)$, then the (random variables)
$\{\sigma(e)\}$ varying over $e\in\Eor_B$ are independent.
\end{definition}
An edge-independent model $\{\cC_n(B)\}_{n\in N}$ is therefore
described by specifying the distribution
of $\sigma(e)\in\cS_n$ for every $n\in N$
and every edge $e\in\Edir_B$, or equivalently, every edge $e\in\Eor_B$
where $\Eor_B\subset\Edir_B$ is some orientation of $B$.

We now describe what we call {\em our basic models}; these models are the
ones that are most convenient for our methods.

\begin{definition}\label{de_models}
Let $B$ be a graph.
By {\em our basic models} we mean one of the models 
edge-independent models $\{\cC_n(B)\}_{n\in N}$ over $B$ of degrees in $N$.
\begin{enumerate}
\item The {\em permutation model} assumes $B$ is any graph without half-loops
and $N=\naturals$: for each $n$ and $e\in\Edir_B$, 
$\sigma(e)\in\cS_n$ is a uniformly chosen permutation.
\item The {\em permutation-involution of even degrees} is defined for any $B$ 
and for $N$ being the even naturals: this is the same as the permutation,
except that if $e$ is a half-loop, then $\sigma(e)$ is a 
uniformly chosen {\em perfect matching} on $[n]$, i.e., a map $\sigma\in\cS_n$
that has no fixed points and satisfies $\sigma^2=\id$.
\item The {\em permutation-involution of odd degrees} is defined the same,
except that $e$ is a half-loop, then $\sigma(e)$ is a
uniformly chosen {\em near perfect matching} on $[n]$, i.e., 
a map $\sigma\in\cS_n$ 
with exactly one fixed point and with $\sigma^2=\id$.
\item
The {\em full cycle} model (or simply {\em cyclic} model)
is defined like the permutation model
(so $B$ is assumed to have no half-loops),
except
that when $e$ is a whole-loop then $\sigma(e)$ is a uniform
permutation whose cyclic structure consists of a single cycle of length $n$.
\item 
The {\em full cycle-involution} (or simply {\em cyclic-involution})
{\em of even degree} and {\em of odd degree} are the two models
defined for arbitrary $B$ and either $n$ even or $n$ odd, is the
full cycle model with the distributions of $\sigma(e)$ for half-loops, $e$,
as in the permutation-involution.
\end{enumerate}
\end{definition}

\subsection{Coincidences and the Order Bound for the Permutation Model}

In this subsection we prove \eqref{eq_algebraic_order_bound}
for 
all of our basic models (Definition~\ref{de_models}).  This
proof is based on the approach of Broder-Shamir to
trace methods for regular graphs; this approach is also the
basis of our asymptotic expansions, which go back to
\cite{friedman_random_graphs}.
We use the notion of {\em coincidences}
of \cite{friedman_random_graphs} (see the second displayed formula on
page~352 for the bound, or Lemma~5.7 of \cite{friedman_alon}), 
which is a straightforward generalization
of Lemma~3 of \cite{broder}; see also Lemma~2.2 of \cite{friedman_relative}.


If $B$ is a graph, then
any walk in an element of $\Coord_n(B)$
is an alternating sequence of vertices and directed edges, and therefore
an alternating sequence of elements of $V_B\times [n]$ and
$\Edir_B\times [n]$; hence such a walk is
necessarily of the form
\begin{equation}\label{eq_walk_in_G}
w
=\bigl( (v_0,i_0), (e_1,i_0), (v_1,i_1),
\ldots,(e_k,i_{k-1}),(v_k,i_k) \bigr)\ ,
\end{equation}
with $i_0,\ldots,i_k\in[n]$ and
$(v_0,\ldots,e_k,v_k)$ an alternating sequence of elements of $V_B$ and 
$\Edir_B$ which we easily verify is a walk in $B$
(via \eqref{eq_coord_def} and \eqref{eq_coord_def_graph});
if $\sigma\from \Edir_B\to\cS_n$ is the map associated to any
$G\in\Coord_n(B)$
(see \eqref{eq_coord_def}),
then $w$ above lies in $G$
iff for $j=1,\ldots,k$ we have
\begin{equation}\label{eq_i_j_trajectory}
i_j 
=\sigma(e_j) i_{j-1} .
\end{equation}
Furthermore, 
\eqref{eq_coord_def} and \eqref{eq_coord_def_graph} easily show that
$w$ above is SNBC iff $i_k=i_0$ and the walk
$(v_0,e_1,\ldots,e_k,v_k)$ is SNBC in $B$.

\begin{definition}\label{de_coincidences}
Let $\pi\from G\to B$ be a coordinatized covering map, $\sigma$
its associated map $\Edir_B\to\cS_n$.
For $w_B=(v_0,\ldots,e_k,v_k)$ and $i_0\in [n]$, 
for $j=0,\ldots,k$ let $i_j=i_j(\sigma,w_B,i_0)$ be as 
in \eqref{eq_i_j_trajectory} and let
\begin{equation}\label{eq_gamma_j}
\gamma_j=\gamma_j(\sigma,w_B,i_0)
\eqdef
\ViSu\Bigl( \bigl(v_0,i_0\bigl), \bigl(e_1,i_0\bigr), 
\ldots,\bigl(e_j,i_{j-1}\bigr),
\bigl(v_j,i_j\bigr) \Bigr) .
\end{equation}
We say that (with respect to $\sigma,i_0,w_B$) $j\in[k]$ is
\begin{enumerate}
\item a {\em forced choice} if $\gamma_j(\sigma)=\gamma_{j-1}(\sigma)$
(i.e., $\sigma(e_j)i_{j-1}$ has already be determined, i.e.,
for some $\ell<j$ either $e_j=e_\ell$ and $i_j=i_\ell$ or
$e_j=\iota_B e_\ell$ and $i_j=i_{\ell+1}$),
and
\item a {\em free choice} otherwise (i.e., $\sigma(e_j)i_{j-1}$
has not been determined by the values of $\sigma(e_\ell)i_\ell$ for 
$\ell<j$, i.e., the edge $(e_j,i_j)$ does not lie in $\gamma_{j-1}$),
and in this case
\begin{enumerate}
\item a {\em coincidence} if $\gamma_j(\sigma)$ has one more edge but the same 
number of vertices as $\gamma_j(\sigma)$, (i.e., $(v_j,i_j)$ lies in
$\gamma_{j-1}$), 
and
\item a {\em new choice}
if $\gamma_j(\sigma)$ has one more edge and one more vertex
than $\gamma_j(\sigma)=\gamma_{j-1}(\sigma)$ (i.e., $(v_j,i_j)$ does not
lie in $\gamma_{j-1}$).
\end{enumerate}
\end{enumerate}
\end{definition}
The terms {\em forced/free choice} is from \cite{broder} (page~289,
end of second paragraph before Lemma~3), and
{\em coincidence} from \cite{friedman_random_graphs} (bottom of page~335).

In other words, we view $\gamma_0,\gamma_1,\ldots,\gamma_k$ as random
graphs that evolve, beginning with $\gamma_0$ which consists of
just $(v_0,i_0)$, ending with $\gamma_k$ which is the entire visited
subgraph of $w$; for each $j\in[k]$, $\gamma_j$ either equals
$\gamma_{j-1}$ (when $\sigma(e_j)i_{j-1}$ has already been determined),
or else $\gamma_j$ consists of one new edge and possibly one new vertex.
Notice that the order of $\gamma_j$ equals the order of $\gamma_{j-1}$
except when $j$ is a coincidence, in which case 
the order of $\gamma_j$ is
one more than that of $\gamma_{j-1}$.  Hence the order $\gamma_k$ is the order
of $\gamma_0$ (i.e., $-1$) plus the number of coincidences among the
$j\in [k]$.

Notice that coincidences and forced/free choices can be viewed as purely
graph theoretic properties of the successive visited subgraphs of 
the first $j$ steps of the
walk in a graph, $j=0,\ldots,k$.

\begin{lemma}\label{le_order_bound}
Let $\cC_n(B)$ be any of our standard models (Definition~\ref{de_models}).
Then $\cC_n(B)$ satisfies \eqref{eq_algebraic_order_bound}, i.e., the
order bound.
\end{lemma}
\begin{proof}
If $w$ is an SNBC walk in a graph $G\in\cC_n(B)$, then $w$ is of the
form \eqref{eq_walk_in_G}, where $w_B=(v_0,\ldots,e_k,v_k)$ is
SNBC in $B$.  
Fix any such $w_B$ and an $i_0\in[n]$.
Consider the event that $\sigma\from \Edir_B\to\cS_n$ is such that
the resulting walk \eqref{eq_walk_in_G} given with
$i_j$ as in \eqref{eq_i_j_trajectory} is SNBC and has order at least $r$,
such a $\sigma$ has at least $r+1$ coincidences.
So fix any $r+1$ values, $j_1<\cdots < j_{r+1}$ in $[k]$ which are the
first $r+1$ coincidences (in fact, any particular $r+1$ coincidences
chosen for each $G\in\cC_n(B)$ would also work).
The probability that a fixed $j\in [k]$ is a coincidence given a
fixed value of $\gamma_{j-1}$
is at most $j/(n-2j+2)$, since at most $2j-2$ values of $\sigma(e_j)$ can
be determined by $\gamma_{j-1}$, and the coincidence happens when
$\sigma(e_j)i_{j-1}$ takes on one of at most $j$ values of at least
$n-2j+2$ possible values in a
uniformly 
chosen permutation or cycle or perfect matching or near perfect matching.
Hence the probability that $j_1<\cdots < j_{r+1}$ are all coincidences is
at most $k/(n-2k+1)^{r+1}$.  Since the number of
choices for $w_B$, for $i_0$, and 
for $1\le j_1<\cdots < j_{r+1}\le k$ are respectively
$$
\Trace(H_B^k), \quad n, \quad \binom{k}{r+1},
$$
the union bound implies that the expected number of SNBC walks of length
$k$ and order at least $r$ is bounded by
\begin{equation}\label{eq_first_order_upper}
\Trace(H_B^k)n\binom{k}{r+1}\left( \frac{k}{n-2k+1}\right)^{r+1}.
\end{equation}
Using the crude bounds
$$
\Trace(H_B^k) \le (\#\Edir_B)\mu_1^k(B), \quad
\binom{k}{r+1} \le k^{r+1},
$$
and, under the assumption that $2k\le n/2$, the bound
$$
\left( \frac{k}{n-2k+1}\right)^{r+1} \le k^{r+1}(n/2)^{-r-1} ,
$$
gives an upper bound on \eqref{eq_first_order_upper} of
$$
O(1) \mu_1^k(B) k^{2r+2}/n^r
$$
(where $O(1)$ depends only on $r$), and is therefore bounded by
$g(k)/n^r$ where $g$ is a function of growth $\mu_1(B)$.
\end{proof}

\subsection{The Permutation Model is Strongly Algebraic}

We now prove that the permutation model is strongly algebraic; 
of all of our basic models, the permutation model
involves
the simplest formulas; all of our other basic models
will be proved to be algebraic or strongly algebraic
in a similar fashion in the remaining subsections.

\begin{lemma}\label{le_expansions_permutation_model}
Let $B$ be a graph without half-loops, and $\cC_n(B)$ the permutation
model over $B$.  If $S_\Bg^\og$ is any ordered \'etale $B$-graph, and
$\Eor_B\subset\Edir_B$ is any orientation of $B$, then
\begin{align}
\label{eq_expected_occur_etale}
\EE_{G\in\cC_n(B)}\bigl[ \#[S_\Bg^\og]\cap G \bigr]
 &  =
\prod_{e\in \Eor_B} \frac{1}{n(n-1)\ldots\bigl(n-a_{S_\Bg}(e)\bigr)} \\
\label{eq_vertex_symmetry}
& \times
\prod_{v\in V_B} \Bigl( n(n-1)\ldots\bigl(n-b_{S_\Bg}(v)+1\bigr) \Bigr) 
\end{align}
(with $\mec a_{S_\Bg},\mec b_{S_\Bg}$ as in 
Subsection~\ref{su_B_ordered_strongly_alg})
provided that $\#V_S,\#E_S\le n$;
if $S_\Bg^\og$ is any ordered $B$-graph that is not \'etale, then
\begin{equation}\label{eq_expected_occur_not_etale}
\EE_{G\in\cC_n(B)}\bigl[ \#[S_\Bg^\og]\cap G \bigr]
=
0.
\end{equation}
\end{lemma}
\begin{proof}
Clearly inclusions and covering maps are \'etale, and clearly the
composition of \'etale maps is \'etale.
Hence any $B$-subgraph of a $G\in\Coord_n(B)$ is an \'etale $B$-graph,
and this implies
\eqref{eq_expected_occur_not_etale} if $S_\Bg$ is not an \'etale
$B$-graph.

It remains to let $S_\Bg^\og$ be \'etale and to
prove \eqref{eq_expected_occur_etale} and \eqref{eq_vertex_symmetry};
let $\pi\from S\to B$ be the structure map of
$S_\Bg$.
Consider any $(S')^\og_\Bg$ that is an element of $S_\Bg^\og\cap G_\Bg$
for some $G\in\cC_n(B)$; then we have
$$
V_{S'} = \bigcup_{v\in V_B} \{v\}\times I_v, 
\quad
\Edir_{S'} = \bigcup_{e\in \Edir_B} \{e\}\times I_e, 
$$
where the above $I_v$ and $I_e$ are subsets of $[n]$, and
$$
t_{S'}(e,i)=(t_B e,i),\quad
h_{S'}(e,i) = \bigl(h_B e,\sigma'(e) i \bigr) ,\quad
\iota_{S'}(e,i) = \bigl(\iota_B e, \sigma'(e)i \bigr) ,
$$
where for each $e\in\Edir_B$, $\sigma'(e)\from I(e)\to I(\iota_B e)$ is
an isomorphism; since $\iota_{S'}(e,i)$ is an involution,
we must have $\sigma'(e)^{-1}=\sigma(\iota_B e)$ for all $e\in\Edir_B$.
The unique isomorphism $S_\Bg^\og\to (S')^\og_\Bg$ gives rise to an
isomorphism for each $v\in V_B$:
\begin{equation}\label{eq_mu_vs}
\mu_v \from \pi^{-1}(v) \to I_v \subset [n].
\end{equation} 
Conversely, we easily see that any other family of injections
$$
\mu''_v \from \pi^{-1}(v) \to I''_v \subset [n]
$$
gives rise (using the heads and tails maps of $S$) to a unique
ordered graph, $(S'')^\og_\Bg$, also isomorphic to $S_\Bg^\og$, and
$(S'')^\og_\Bg$ and $(S')^\og_\Bg$ are isomorphic iff
$\mu''_v=\mu_v$ for all $v$.
(Here the orderings are crucial, since $S''_\Bg$ can be isomorphic
as a $B$-graph to $S_\Bg$ without $(S'')^\og_\Bg$ and $(S')^\og_\Bg$
being isomorphic.)

Since $|I_v|=b_{S_\Bg}(v)$, the number of families $\{\mu_v\}_{V_B}$ 
of injections as in \eqref{eq_mu_vs} is
$$
\prod_{v\in V_B} \Bigl( n(n-1)\ldots\bigl(n-b_{S_\Bg}(v)+1\bigr) \Bigr) .
$$
Furthermore, clearly a $G\in\cC_n(B)$, with corresponding permutation map
$\sigma$, contains $S'_\Bg$ as above iff
for each $e\in\Eor_B$, 
$\sigma\in\cS_n$ agrees with $\sigma'$ on $I_e$; for each $e\in\Eor_B$
this occurs with probability
$$
\frac{1}{n(n-1)\ldots(n-|I_e|+1}.
$$
Since $|I_e|=a_{S_\Bg}(e)$, we conclude
\eqref{eq_expected_occur_etale} and \eqref{eq_vertex_symmetry}.

\end{proof}

\begin{lemma}\label{le_permutation_is_alg}
Let $B$ be a graph without half-loops, and $\cC_n(B)$ the permutation
model over $B$.  Then $\cC_n(B)$ is strongly algebraic.
\end{lemma}
\begin{proof}
According to Lemma~\ref{le_order_bound}, $\cC_n(B)$ satisfies
\eqref{eq_algebraic_order_bound}.
According to Lemma~2.8 of \cite{friedman_random_graphs} we have that
for fixed integers $a\ge 0$ and $r>1$ we have
\begin{equation}\label{eq_a_factor_expansion}
\frac{1}{n(n-1)\ldots (n-a+1)} =
n^{-a}\bigl( 1 + R_1(a)n^{-1} + \cdots + R_{r-1}(a) n^{-r+1} 
+ O(n^{-r}) \bigr),
\end{equation} 
where the $R_i(a)$ are polynomials of degree $2i$.
Lemma~2.9 there shows a similar expansion
$$
n(n-1)\ldots (n-b+1)=
n^b \bigl( 1 - Q_1(b) n^{-1} + \cdots + (-1)^{r-1} Q_{r-1}(b) n^{-r+1}
+ O(n^{-r}) \bigr)
$$
where the $Q_i(b)$ are polynomials of degree $2i$.  It follows that
\begin{equation}\label{eq_power_series}
\EE_{G\in\cC_n(B)}\bigl[ \#([S_\Bg^\og] \cap G)  \bigr]
=
n^{-\ord(S)}
\bigl(1 + c_1(\mec a,\mec b)n^{-1} + \cdots + 
c_{r-1}(\mec a,\mec b) n^{-r+1} + \epsilon n^{-r} \bigr),
\end{equation} 
where $c_i$ is a polynomial of degree $2i$, and---according to
Lemma~2.7 of 
\cite{friedman_random_graphs} and (6)
there (see also the discussion around (20) in
\cite{friedman_alon})---where $\epsilon=\epsilon(n,\mec a,\mec b,r)$
is bounded by
\begin{equation}\label{eq_alpha_beta_bound}
|\epsilon(n,\mec a,\mec b,r)| \le
\frac{1}{(1-C/n)^r} ( \alpha+\beta)^r
\end{equation}
where
$$
\alpha =
\sum_{e\in E_B} \Bigl(1+2+\ldots+ \bigl(a_{S_\Bg}(e)-1\bigr)\Bigr) 
$$
and
$$
\beta =
\sum_{v\in V_B} \Bigl(1+2+\ldots+ \bigl(b_{S_\Bg}(v)-1\bigr)\Bigr) 
$$
and where $C$ is an upper bound on the components of $\mec a$ and of $\mec b$;
we may take $\#E_S$ as such an upper bound on these components, and we
easily check that
$$
\alpha = \sum_{e\in E_B} \Bigl(1+2+\ldots+ \bigl(a_{S_\Bg}(e)-1\bigr)\Bigr)
=\sum_{e\in E_B}\binom{a_{S_\Bg}(e)}{2}
\le \binom{\mec a\cdot\mec 1}{2} = \binom{\#E_S}{2} \le (\#E_S)^2
$$
and we similarly bound 
$$
\beta\le (\#V_S)^2 \le (\#\Edir_S)^2
$$
(the inequality $\#V_S\le\#\Edir_S$ follows since $S$ has no isolated vertices
and hence each vertex of $S$ is the tail of some directed edge of $S$);
it follows that
for $\#E_S\le n/2$,
\begin{equation}\label{eq_error_term}
|\epsilon(n,\mec a,\mec b,r)| \le O(\#E_S)^{2r} \ .
\end{equation} 
This establishes 
\eqref{eq_expansion_S} and \eqref{eq_strongly_algebraic}
in the case where $S_\Bg^\og$ is an \'etale $B$-graph; if $S_\Bg^\og$
is not \'etale, then Lemma~\ref{le_expansions_permutation_model} shows
that the left-hand-side of \eqref{eq_expansion_S} vanishes, whereupon
one can take $c_0=\cdots=c_{r-1}=0$ to satisfy 
\eqref{eq_expansion_S}.

It remains to show that if $S$ is a cycle and $S_\Bg$ is an
\'etale $B$-graph, then $c_0(S_\Bg)=1$;
but this follows from \eqref{eq_power_series}.
\end{proof}


\subsection{The Permutation-Involution Model of Even Degree is Strongly
Algebraic}

\begin{lemma}\label{le_even_invol_strongly}
Let $B$ be a graph with half-loops, and let $\{\cC_n(B)\}_{n\in N}$ be the 
permutation-involution model of even degree (so $N$ consists of the even
natural numbers).
Then $\cC_n(B)$ is strongly algebraic.
\end{lemma}
\begin{proof}
This follows from the proof of Lemma~\ref{le_permutation_is_alg}.  
The only difference is that if $e\in \Edir_B$ is a half-loop, then
$\sigma(e)\in\cS_n$ is required to be an involution without fixed points,
and so if $e$ occurs $a$ times in $S_\Bg^\og$, the probability
that an $(S')_\Bg^\og$ occurs as a subgraph of $\cC_n(B)$ is 
$$
\frac{1}{ (n-1)(n-3)\ldots (n-2a+1)} \ .
$$
Hence 
the probability of $(S')_\Bg^\og$ being contained
in an element of $\cC_n(B)$ is
\begin{equation}\label{eq_perm_invol_prob}
\prod_{e\in \Eor_B\setminus {\rm Half}_B} 
\frac{1}{n(n-1)\ldots \bigl(n-a_{S_\Bg}(e) + 1 \bigr)}
\prod_{e\in {\rm Half}_B} 
\frac{1}{(n-1)(n-3)\ldots \bigl(n-2 a_{S_\Bg}(e) + 1 \bigr)},
\end{equation} 
where ${\rm Half}_B$ denotes all the half-edges of $B$ and 
$\Eor_B\subset\Edir_B$ is
an orientation of $B$.
Hence we get an asymptotic expansion of this probability in powers
of $1/n$, with different polynomials $p_i=p_i(\mec a,\mec b)$ reflecting
the fact that for fixed $a$
$$
\frac{1}{(n-1)(n-3)\ldots \bigl(n-2 a+ 1 \bigr)}
$$
has coefficients that are different polynomials in $a$,
but whose leading term is still $n^{-a}c_0$ with $c_0=1$.
\end{proof}

\subsection{Strongly Algebraic Models are Algebraic}

Here we formally state the almost immediate fact that a strongly
algebraic model is also algebraic.
\begin{lemma}
Let $B$ be a graph.
Any strongly algebraic model over $B$ is also algebraic,
and a set of eigenvalues for such a model
is the set of eigenvalues of $H_B$.
\end{lemma}
In this proof we say that a map $\pi\from S\to B$ of graphs
is {\em \'etale at} a 
vertex, $v$, of $S$ if $\pi$ is an injection when
restricted to the elements of $\Edir_S$
whose head is $v$; since $S,B$ are graphs,
this condition is equivalent if ``head'' is replaced with
``tail;'' hence, by our definitions,
$\pi$ is \'etale iff it is \'etale at each vertex of $S$.
We similarly speak of a $B$-graph, $S_\Bg$, as being {\em \'etale at}
a vertex of $S$, referring to the structure map $S\to B$.
\begin{proof}
The coefficients for the asymptotic expansions for
$$
f(k,n) = 
\EE_{G\in\cC_n(B)}\Bigl[ \#\bigl([S_\Bg^\og]\cap G\bigr) \Bigr]
$$
depend on whether or not $S_\Bg^\og$ is an \'etale $B$-graph.
If $S_\Bg^\og$ is \'etale, then $f(k,n)$ above has an expansion
with coefficients
$$
c_{\ord(S)+i}(S_\Bg) = p_i(\mec a_{S_\Bg},\mec b_{S_\Bg}) 
$$
(and if $S_\Bg^\og$ is not \'etale, then all coefficients vanish).
But if $S_\Bg^\og$ is of homotopy type $T^\og$,
then $\ord(S)=\ord(T)$ and 
$\mec b_{S_\Bg}$ is determined by $\mec a_{S_\Bg}$ and $\ord(T)$.
Hence the coefficients are polynomials of $\mec a_{S_\Bg}$ alone
when $S_\Bg^\og$ is of a fixed homotopy type.

To prove the theorem it therefore remains to see fix an ordered graph,
$T^\og$, and prove the following:
we can subdivide all $B$-wordings, $W$, of $T$ into types---i.e.,
expressed as regular languages associated to each $e\in\Edir_T$, 
that express
the property that $S^\og_\Bg=\VLG_\Bg^\og(T^\og,W)$ is
\'etale.  We easily see that for such a $W$ and $S^\og_\Bg$,
the map $S\to B$ is \'etale at
each vertex of $S$ that is not a vertex of $T$ (i.e., each vertex of
$S$ that is an intermediate vertex in a beaded path associated to some $W(e)$),
and for $v\in V_T\subset V_S$, $S^\og_\Bg$ is \'etale at $v$ iff
the edges in $\Edir_S$ whose tail is $v$ are mapped to distinct edges
in $\Edir_B$.
But this latter property depends only on the first and last letters
of $W(e)$ for all $e\in\Edir_T$; moreover, the set of $W(e)$ that begin and
end with, respectively, $e_1,e_2\in\Edir_B$, is a regular language,
and the eigenvalues of this regular language
are a subset of the eigenvalues of $H_B$, since the number of such words
of length $k$ is the $e_1,e_2$ entry of $H_B^k$.
\end{proof}
The knowledge of the letters with which each $W(e)$ begins and ends
was called the
{\em lettering} in \cite{friedman_random_graphs,friedman_alon}.

\subsection{The Permutation-Involution Model of Odd Degree is Algebraic}

\begin{lemma}\label{le_odd_invol_strongly}
Let $B$ be a graph with half-loops, and let $\{\cC_n(B)\}_{n\in N}$ be the 
permutation-involution model of odd degree (so $N$ consists of the odd
natural numbers).
Then $\cC_n(B)$ is algebraic, and a set of eigenvalues for this model
is the set of eigenvalues of $H_B$.
\end{lemma}
\begin{proof}
This follows from the proof of Lemmas~\ref{le_permutation_is_alg}
and \ref{le_even_invol_strongly}.
The main difference is that there are two types of maps
$\sigma'(e)\from I(e)\to [n]$ over half-loops $e\in \Edir_B$, which are
required to be involutions, namely
\begin{enumerate}
\item those where $\sigma'(e)$ has not specified the unique fixed point
of the involution, so that if $a$ edges over $e$ occur in $S_\Bg^\og$ then
the probability that any $(S')_\Bg$ is a subgraph of an element of
$\cC_n(B)$ is
\begin{equation}\label{eq_no_fixed_pt_specified}
\frac{1}{(n-2)(n-4)\ldots(n-2a)},
\end{equation} 
and
\item otherwise $S_\Bg$ and $S'_\Bg$ have a half-loop over this $e$,
and the probability becomes
\begin{equation}\label{eq_fixed_pt_specified}
\frac{1}{n(n-2)\ldots (n-2a+2)}.
\end{equation} 
\end{enumerate}
Furthermore the reduction of $S_\Bg^\og$ 
contains $e$ and its incident vertex $v=t_S e = h_S e$;
therefore if $\pi\from (S')_\Bg^\og\to\VLG_\Bg^\og(T^\og,W)$ is an 
isomorphism, then $W(\pi(e))$ gives the edge in $\Edir_B$---which is
necessarily a half-loop---over
which $\pi(e)$ lies.
Hence knowing the homotopy type of $T$ and the first letter of each
$W(e)$ with $e\in\Edir_T$ allows us to infer which of the two above cases
applies to each half-loop, $e$.

Hence for a fixed ordered graph, $T^\og$, the coefficients of the 
asymptotic expansion of 
$$
f(k,n) = \EE_{G\in\cC_n(B)}\Bigl[ \#\bigl([S_\Bg^\og]\cap G\bigr) \Bigr]
$$
of any $S_\Bg^\og$ isomorphic to an ordered $B$-graph of the form
$\VLG_\Bg^\og(T^\og,W)$ depend on knowing only the first and last letters
of $W(e)$ for all $e\in\Edir_B$.
Since both 
\eqref{eq_no_fixed_pt_specified} and
\eqref{eq_fixed_pt_specified} have leading term $n^{-a}c_0$ with
$c_0=1$, we again verify 
\eqref{eq_expansion_S} and \eqref{eq_strongly_algebraic};
as for all our basic models,
\eqref{eq_algebraic_order_bound} follows
from Lemma~\ref{le_order_bound}.
\end{proof}

\subsection{The Cyclic Model is Algebraic}

\begin{lemma}\label{le_cyclic_is_algebraic}
Let $B$ be a graph without half-loops, and let $\{\cC_n(B)\}_{n\in \naturals}$ 
be the cyclic model.
Then $\cC_n(B)$ is algebraic, and 
a set of eigenvalues for this model consists of the eigenvalues of
$H_B$ and possibly $1$.
\end{lemma}
\begin{proof}
Our proof is again based on that of Lemma~\ref{le_permutation_is_alg}, but
there is one essential difference.  For an $n\in\naturals$,
an $I\subset [n]$, and a map $\sigma'\from I\to [n]$, say that 
$\sigma'$ is {\em feasible} if the following
evidently equivalent conditions hold:
\begin{enumerate}
\item there is a $\sigma\in\cS_n$ such that $\sigma$ is a full cycle
whose restriction, $\sigma|_I$, to $I$ equals $\sigma'$;
\item the directed graph $G=G_{\sigma'}$ given by
\begin{equation}\label{eq_feasibility_graph_equation}
V_G=[n],\  \Edir_G=I, \ t_G={\rm Identity}, h_G=\sigma'
\end{equation} 
(which is necessarily of indegree and outdegree at most one at each vertex)
has no cycles of length less than $n$.
\end{enumerate} 
We easily see that
if $\sigma'\from I\to [n]$ is feasible, then
for any $i\notin I$,
we may extend $\sigma'$ to a feasible map $I\cup\{i\}\to[n]$
in $n-a-1$ ways if $a=\#I<n$ (i.e., there are $n-a-1$ possible values for
the new value $\sigma'(i)$); therefore the number of full cycles
that agree with $\sigma'$ on $I$ is $(n-a-1)!$.
It follows that if $\sigma'$ is feasible, the probability that a 
random full cycle, $\sigma\in\cS_n$, agrees with $\sigma'$ on $I$ is
\begin{equation}\label{eq_full_cycle_prob}
\frac{(n-a-1)!}{(n-1)!} = \frac{1}{(n-1)(n-2)\ldots (n-a)}
\end{equation} 
Like the similar probability expressions involving $n$ and $a$, this
function also has an asymptotic expansion to any order with leading
term $n^{-a}c_0$ with $c_0=1$ and coefficients that are polynomials in 
$a$.

The subtlety is that for a ordered graph $T^\og$, we need to know
which $B$-wordings, $W$, are {\em feasible} in the sense
that 
for any $(S')_\Bg^\og$ isomorphic to $\VLG_\Bg^\og(T,W)$, the associated
map $\sigma'\from\Edir_{S'}\to\cS_n$ to $S'_\Bg$
has maps $\sigma'(e)\from I(e)\to [n]$ that are feasible
(for all whole-loops, $e\in\Edir_B$).
(And we must describe such wordings in terms of regular languages.)
Since the formulas and expansions we obtain need hold 
only for $(S')_\Bg^\og$
with $\#\Edir_{S'}\le n^{1/2}/C$ for a constant, $C$
(of our choosing, for fixed $T^\og$), we may always assume
$C>1$, so that
$\#\Edir_{S'}\le n-1$.
Hence, if for any whole-loop, $e\in \Edir_B$, the graph in
\eqref{eq_feasibility_graph_equation} has a cycle---with
$\sigma'=\sigma'(e)$ and $I=I(e)$---then
this cycle is automatically of length strictly less than $n$.
It follows that the
$B$-wordings, $W$, that are feasible in this sense
are precisely those for which
for any whole-loop, $e\in\Edir_B$, the edges over $e$ in
$\VLG_\Bg^\og(T,W)$ have no cycle.
So to determine the correct polynomial of $\mec a,\mec b$ for each
coefficient $c_i(S_\Bg)$, we not only need to know
the first and last letter of each $W(e_T)$ for $e_T\in\Edir_T$, 
but also which of the words $W(e_T)$ is a power $e_B^k$ for some
whole-loop $e_B\in\Edir_B$.
So it suffices to know if $W(e_T)$ lies in the regular language
$e_B^*$ for some $e_B\in\Edir_T$ with $e_B$ a whole-loop, or in
the languages
\begin{equation}\label{eq_cyclic_model_regular_langs}
{\rm NBWALKS}(B,e',e'') \setminus \bigcup_e \{ e^* \}
\end{equation} 
with $e',e''$ ranging over $\Edir_B$ and $e$ above ranging over all
whole-loops.
Since the language $\{ e^* \}$ has exactly one word of each length,
its sole eigenvalue is $1$;
hence the eigenvalues of the regular language $\{ e^* \}$ or
any language of the form \eqref{eq_cyclic_model_regular_langs}
are those of $H_B$ and (possibly) $1$.
\end{proof}
The above subtlety regarding the cycle model was overlooked in
\cite{friedman_alon}.

\subsection{All Our Basic Models Are Algebraic}

We are now able to finish our claims about all our basic models.

\begin{lemma}
Let $B$ be a graph.  All our basic models are algebraic,
and a set of eigenvalues for each model
consist of possibly $1$ and some subset of
the eigenvalues $\mu_i(B)$ of the Hashimoto matrix $H_B$.
\end{lemma}
\begin{proof}
The proof that the cyclic-involution models of even and of odd degree
follows by combining the proofs of the permutation-involution and
cyclic models above.
The regular languages used in the types are all of one of the three forms:
$$
{\rm NBWALKS}(B,e',e''),
e^*,
{\rm NBWALKS}(B,e',e'') \setminus\bigcup_e e^*;
$$
the exponents of the first form are some subset of the $\mu_i(B)$,
since the number of words of length $k$ in these languages is an
entry of $H_B^k$; the exponents of the language $e^*$ is $1$ since
there is exactly one word of length $k$ for each $k$;
and the number of words of length $k$ in 
${\rm NBWALKS}(B,e',e'') \setminus\bigcup_e e^*$ is that
of ${\rm NBWALKS}(B,e',e'')$ unless $e'=e''=e$ is a whole-loop of $B$,
in which case the number is the same minus $1$.
\end{proof}

\section{The Proof of the Relativized Alon Conjecture for Regular Base Graphs}
\label{se_proof_first}


In this section we prove Theorem~\ref{th_rel_Alon_regular2}
and then Theorem~\ref{th_first_main_thm},
the Relativized Alon Conjecture for regular base graphs.

\subsection{Main Lemma} 

Our first lemma is an immediate consequence of 
Theorem~\ref{th_main_tech_result}.
\begin{lemma}\label{le_corr_main_tech}
Let $B$ be a connected graph with $\mu_1(B)>1$, and let
$\{\cC_n(B)\}_{n\in N}$ be
an algebraic model over $B$.
Let $r>0$ be an integer and $\nu\ge\mu_1^{1/2}(B)$ be a real number.
Then 
\begin{equation}\label{eq_f_k_n_nu_r}
f(k,n)=f_{\nu,r}(k,n)\eqdef
\EE_{G\in\cC_n(B)}\Bigl[
\Bigl( \Trace\bigl(H_G^k\bigr) - \Trace\bigl(H_B^k\bigr) \Bigr)
\II_{{\rm TangleFree}(\ge\nu,<r)}(G)
\Bigr]
\end{equation} 
has a $(B,\nu)$-bounded expansion to order $r$
\begin{equation}\label{eq_new_c_i}
c_0(k)+\cdots+c_{r-1}(k)/n^{r-1}+O(c_{r}(k))/n^{r},
\end{equation} 
such that
\begin{enumerate}
\item
$c_0(k)$ is of growth $(d-1)^{1/2}$
(and is independent of $\nu$ and $r$);
\item
the larger exponent
bases of each $c_{i}(k)$ (with respect to $\nu$)
is some subset of the union of the eigenvalues of $H_B$ and the eigenvalues
of the model.
\end{enumerate}
\end{lemma}
Notice that in the above lemma, $f(k,n)$ and the $c_i(k)$ all depend
on $\nu,r$.
\begin{proof}
The $c_{i}(k)$ in \eqref{eq_new_c_i} equal, in the notion of
Theorem~\ref{th_main_tech_result},
$$
c_i(k) - \widetilde c_i \Trace\bigl(H_B^k\bigr) .
$$
The claim about $c_0(k)$ in \eqref{eq_new_c_i} follows from the
fact that $\widetilde c_i=1$ and \eqref{eq_c_zero_roughly}.
\end{proof}

\subsection{Proof of Theorem~\ref{th_rel_Alon_regular2}}

\begin{proof}[Proof of Theorem~\ref{th_rel_Alon_regular2}]

Let $q=\#\Edir_B$.
For a $G\in\Coord_n(B)$, let $\tilde H_G$ denote the restriction of
$H_G$ to the new space of functions $\Edir_G\to\reals$
(i.e., whose sum on each $\Edir_B$ fibre is zero);
then $\tilde H_G$ can be viewed as a $(n-1)q\times(n-1)q$ square matrix
with respect to some basis of the new functions, and 
$\Trace(\tilde H_G^k)$ is independent of this basis.

Let
\begin{align*}
f(k,n)
 & =
\EE_{G\in\cC_n(B)}\bigl[ \II_{{\rm TangleFree}(\ge\nu,<r)}(G)
\Trace(H_G^k-H_B^k) \bigr] \\
  & =
\EE_{G\in\cC_n(B)}\bigl[ \II_{{\rm TangleFree}(\ge\nu,<r)}(G)
\Trace(\tilde H_G^k) \bigr] .
\end{align*}
Then
$$
f(k,n) = \EE_{M\in \cM_{(n-1)q}}[\Trace(M^k)],
$$
where $\cM_{(n-1)q}$ is the space of random matrices
$$
\II_{{\rm TangleFree}(\ge\nu,<r)}(G)
\tilde H_G
$$
where $G$ varies over $\cC_n(B)$.
Setting $\Lambda_0=\nu$ and $\Lambda_1=d-1$,
Theorem~\ref{th_main_tech_result} shows that
$\{\cM_{(n-1)q}\}_{n\in N}$ is a 
$(\Lambda_0,\Lambda_1)$-matrix model.  Hence, by
Theorem~\ref{th_sidestep}, for sufficiently small
$\epsilon''>0$, either
\begin{equation}\label{eq_Lambda_zero_eps_first}
\EE{\rm out}\bigl[ B_{\Lambda_0+\epsilon''}(0) \bigr]
= O(n^{-j}) 
\end{equation} 
for all $j$, or else for some $\tau\in\naturals$ we have that for
sufficiently large $n$
\begin{equation}\label{eq_Lambda_zero_eps_second}
C'n^{-\tau}\le
\EE{\rm out}\bigl[ B_{\Lambda_0+\epsilon''}(0) \bigr]
\le C(\epsilon'') n^{-\tau}
\end{equation} 

The Ihara Determinantal Formula implies that each eigenvalue
$\lambda$ of $A_G$ of a $d$-regular graph $G$ corresponds to two eigenvalues
$\mu$ of $H_G$ given by
$$
\mu^2 - \lambda \mu + (d-1) = 0
$$
(and aside from these $2n$ eigenvalues of $H_G$, the other eigenvalue
of $H_G$ are $\pm 1$).
In particular, there is a one-to-one correspondence between
eigenvalues of $H_G$, $\mu$, with $|\mu|>(d-1)^{1/2}$ and those
eigenvalues of $A_G$, $\lambda$, with $|\lambda|>2(d-1)^{1/2}$,
taking $\mu$ to
$$
\lambda=\mu+\frac{d-1}{\mu}.
$$
In particular, since $\Lambda_0=\nu>(d-1)^{1/2}$, there is a 
one-to-one correspondence between $H_G$ new eigenvalues
outside $B_{\Lambda_0+\epsilon''}(0)$ and $A_G$ new eigenvalues
outside
$$
\lambda_{\epsilon''} \eqdef \nu+\epsilon'' + \frac{d-1}{\nu+\epsilon''}.
$$
Since
$$
\nu+\frac{d-1}{\nu}= 2(d-1)^{1/2} + \epsilon' ,
$$
and since $\mu+(d-1)/\mu$ is continuous and monotone increasing
for $\mu>(d-1)^{1/2}$, for any sufficiently small $\epsilon>0$
there is an $\epsilon''>0$ such that $\lambda_{\epsilon''}$ above
equals $2(d-1)^{1/2} + \epsilon'+\epsilon$.
For this value of $\epsilon''$ we have
$$
\EE{\rm out}_{\cM_{(n-1)q}} \bigl[ B_{\Lambda_0+\epsilon''}(0) \bigr]
=
\EE_{G\in\cC_n(B)}[ \II_{{\rm TangleFree}(\ge\nu,<r)}(G)
{\rm NonAlon}_d(G;\epsilon'+\epsilon) ] .
$$
Then \eqref{eq_Lambda_zero_eps_first} 
and \eqref{eq_Lambda_zero_eps_second} imply
the claim of the theorem.
\end{proof}

 
\section{The Fundamental Subgraph Lemma}
\label{se_fundamental_subgr}

In this section we prove the following
fundamental lemma regarding new spectrum
of graphs that contain fixed subgraph.

\begin{lemma}\label{le_fundamental_subgraph}
Let $d\ge 3$ be an integer,
$B$ a $d$-regular graph, and $\psi_\Bg$ a fixed $B$-graph with 
$\mu_1(\psi)>(d-1)^{1/2}$.
Set
\begin{equation}\label{eq_lambda_fund} 
\lambda = \mu_1(\psi) + \frac{d-1}{\mu_1(\psi)}.
\end{equation} 
Then for any $\epsilon>0$ there exists an $n_0\in\integers$ such that
if $\pi\from G\to B$ is a covering map
of degree at least $n_0$ such that $G_\Bg$ has
a subgraph isomorphic to $\psi_\Bg$, then
$\Specnew_B(A_G)$ contains an eigenvalue larger than 
$\lambda-\epsilon$.
\end{lemma}
We remark that the above lemma holds, more simply, in the context of graphs
rather than $B$-graphs (i.e., it is enough for $G$ to be $d$-regular
and that $\psi$ be any graph);
however, is it simpler to work with $B$-graphs
in the construction
of a universal cover (namely ${\rm Tree}_\Bg(\psi_\Bg)$)
described below.

Lemma~\ref{le_fundamental_subgraph} is a generalization of
the Theorem~3.13 (i.e., the ``Curious Theorem'' of Subsection~3.8) of 
\cite{friedman_alon} which proves a form of the above lemma when
$B$ has one vertex.

Our proof Lemma~\ref{le_fundamental_subgraph} uses two types of methods:
(1) the general methods of Section~8 
(in particular, Theorem~8.2)
of \cite{friedman_tillich_generalized}, which are based on
the {\em first Dirichlet eigenvalue} of a {\em graph with boundary}
Theorem~2.3 of \cite{friedman_geometric_aspects},
and (2) some specific facts regarding graphs with tangles proven in Section~3
of \cite{friedman_alon}.
The first type of methods are a more robust variant of the
methods used for the Alon-Boppana theorem
(\cite{nilli}, and its improvement by Friedman
\cite{friedman_geometric_aspects} Corollary~3.7 and
Kahale \cite{kahale} Section~3); in addition to being
more robust---which is demonstrated in this section---the 
methods of Friedman-Tillich
yield slightly better constants in the error term
of \cite{friedman_geometric_aspects,kahale}.
The second type of methods are based on those used to prove the
``Curious Theorem,'' Theorem~3.13, of \cite{friedman_alon},
which are based on standard types of calculations with forms of
{\em Shannon's algorithm}, and the relation with the spectrum of
infinite graphs and their finite quotients; we cite
Buck \cite{buck} for the relation we need, which is the earliest
explicit reference we know for these results
(this relation seems so fundamental that they may have appear
elsewhere in spectral theory, at least implicitly).

Let us state a number of preliminary definitions and lemmas, mainly
reviewing the above methods, before
we prove Lemma~\ref{le_fundamental_subgraph}.

\subsection{Basic Notation and the Perturbation of Rayleigh Quotients}

In this subsection we give some general notation and facts, including
a standard type of estimate when a Rayleigh quotient argument 
is perturbed (Lemma~\ref{le_Rayleigh_perturb} below).

First we recall that if $A$ is any real, $n\times n$, symmetric matrix
and $f\in\reals^n$ is nonzero, 
then the {\em Rayleigh quotient of $f$ on $A$} is
defined as 
$$
\cR_A(f) \eqdef \frac{(Af,f)}{(f,f)} 
$$
where $(\, , \,)$ denotes the standard inner product on $\reals^n$;
we have
$$
|\cR_A(f)| \le \| A \| \, \|f\|,
$$
where $\|A\|$ is the $L^2$-operator norm of $A$, which equals the
largest absolute value of an eigenvalue of $A$; furthermore, if 
$A$ has non-negative entries and $f\ge 0$ (i.e., pointwise, i.e.,
$f$ has non-negative components), then for $f$ nonzero and $\lambda\in\reals$
we have
\begin{equation}\label{eq_pointwise_implication}
Af\ge \lambda f \implies \cR_A(f)\ge\lambda 
\end{equation} 
(again, $Af\ge \lambda f$ means pointwise, i.e., component-by-component);
to prove \eqref{eq_pointwise_implication} we note that
$$
(Af,f) \ge (\lambda f,f) = \lambda(f,f),
$$
and we divide by $(f,f)$.

Next recall the special case of adjacency matrices of graphs
\cite{friedman_geometric_aspects,friedman_tillich_generalized}
and the following notation and easy facts:
if $G$ is a graph and $f\from V_G\to\reals$, then the
{\em Rayleigh quotient of $f$ (on $G$, or of $A_G$)} is given by
$$
\cR_{A_G}(f) \eqdef \frac{(A_Gf,f)}{(f,f)}
$$
where $(\, , \,)$ denotes the standard inner product on $\reals^{V_G}$ given by
\begin{equation}\label{eq_inner_product}
(f_1,f_2) = (f_1,f_2)_{L^2(V_G)} \eqdef \sum_{v\in V_G} f_1(v) f_2(v).
\end{equation}
Aside from the norm
$$
\| f \|_2 \eqdef \| f \|_{L^2(V_G)} \eqdef \sqrt{ (f,f) } ,
$$
we will have occasion to use the following notation and easy facts:
\begin{enumerate}
\item we have
$$
\|f\|_2 \ge \| f \|_\infty \eqdef \max_{v\in V_G} |f(v)|  
$$
(since each vertex in $L^2(V_G)$ above has measure $1$);
\item we use ${\rm supp}(f)$ to denote the support of $f$
(i.e., the set of vertices on which $f$ does not vanish);
\item we have
$$
\bigl| (f_1,f_2) \bigr| \le \bigl(\# {\rm supp}(f_1 f_2) \bigr)
\| f_1 f_2 \|_\infty \le \bigl(\# {\rm supp}(f_1) \bigr)
\| f_1\|_\infty \| f_2 \|_\infty
$$
which implies
\begin{equation}\label{eq_f_t_ineq}
\bigl| (f_1,f_2) \bigr| \le \bigl(\# {\rm supp}(f_1) \bigr)
\| f_1\|_2 \| f_2 \|_\infty
\end{equation}
since $\| f_1 \|_2\ge \| f_1 \|_\infty$.
\end{enumerate}

We will need an easy perturbation result on the Rayleigh quotient, which
we state in a general context.

\begin{lemma}\label{le_Rayleigh_perturb}
Let $A$ be any bounded, symmetric
operator on a real inner product space.  If $f_1,f_2$ are elements of
the space with
$$
\|f_2-f_1\|\le \epsilon \| f_1 \|
$$
for some $0<\epsilon<1$, then 
\begin{equation}\label{eq_Rayleigh_perturb_result}
|\cR_A(f_2)-\cR_A(f_1)| \le  2 \| A \| \,\epsilon
\end{equation} 
where $\cR_A(f)\eqdef (Af,f)/(f,f)$.  
\end{lemma}
(The above lemma also holds for $\epsilon\ge 1$ for the trivial
reason that $|\lambda_i|\le \|A\|$ for $i=1,2$.)
\begin{proof}
Let us make some simplifying assumptions before making a computation.
It suffices to work
in, $V$, the $2$-dimensional span of $f_1,f_2$.  By choosing an orthonormal
basis for $V$, 
we may assume that we are working in $\reals^2$ under
the standard inner product.  By choosing an orthonormal eigenbasis
for $A$ and applying the associated orthogonal matrix, we may further
assume that
$$
A = \begin{bmatrix} \lambda_1 & 0 \\ 0 & \lambda_2 \end{bmatrix}  ,
$$
and hence $\|A\|=\max(|\lambda_1|,|\lambda_2|)$.
Since $f_1=0$ implies $f_2=0$ and hence implies
\eqref{eq_Rayleigh_perturb_result},
we may assume $f_1\ne 0$.
Since $\cR_A(f)$ is invariant under scaling $f$,
we may divide $f_1$ and $f_2$ by $\|f_1\|$, so that we may
further assume that $f_1$ is a 
unit vector, say $(\cos\theta,\sin\theta)$.
Finally, $f_2\ne 0$ since $\epsilon<1$, and hence
$f_2$ is a positive multiple of a unit vector $(\cos\theta',\sin\theta')$
with $|\theta'-\theta|< \pi/2$,
and the closest multiple of $(\cos\theta',\sin\theta')$ to
$(\cos\theta,\sin\theta)$ is (by drawing a trigonometry diagram)
at a distance $|\sin(\theta-\theta')|$; hence
$$
|\sin(\theta-\theta')| \le \|f_2-f_1\| \le \|f_1\| \epsilon = \epsilon.
$$

Now we make an easy computation.
We have
$$
\cR_A(f_1)=\lambda_1 \cos^2\theta + \lambda_2\sin^2\theta
= (\lambda_1-\lambda_2)\cos^2\theta+ \lambda_2,
$$
and similarly for $\cR_A(f_2)$, whereupon
\begin{equation}\label{eq_Rayleigh_diff}
|\cR_A(f_2)-\cR_A(f_1)| = |\lambda_1-\lambda_2| \,
|\cos^2\theta-\cos^2\theta'|.
\end{equation} 
Using the fact that
$$
\cos^2\theta-\cos^2\theta' = \sin(\theta+\theta')\sin(\theta-\theta')
$$
it follows that
$$
|\cos^2\theta-\cos^2\theta'| \le |\sin(\theta-\theta')|\le\epsilon ;
$$
hence this bound and \eqref{eq_Rayleigh_diff} implies that
$$
|\cR_A(f_2)-\cR_A(f_1)| \le |\lambda_1-\lambda_2| \epsilon
\le (|\lambda_1|+|\lambda_2|) \epsilon \le 2\|A \| \epsilon.
$$
%
\end{proof}

\subsection{Friedman-Tillich Methods of Large Graphs Containing
a Fixed Subgraph}

In this subsection we describe some 
results akin to those of Section~8 of \cite{friedman_tillich_generalized},
partially inspired by \cite{friedman_geometric_aspects}.

\begin{definition}\label{de_extension_zero}
Let $\psi\subset G$ be graphs,
and let $f\from V_\psi\to\reals$ be a function.
By the {\em extension of $f$ by zero in $G$}, denoted $f_G$, we mean
the function
$V_G\to\reals$ that is $f$ on $V_\psi\subset V_G$
and otherwise $0$ (i.e., on $V_G\setminus V_\psi$).
\end{definition}

\begin{lemma}\label{le_friedman_tillich}
Let $\psi$ be a graph and $f$ be a function $V_\psi\to\reals$.
Let $G$ be any graph containing $\psi$ as a subgraph,
and for some $m\in\naturals$ and $\delta>0$
let $f_1,\ldots,f_m$ be a set of orthonormal
functions on $V_G$ (with respect to \eqref{eq_inner_product})
such that $\|f_i\|_\infty \le \delta$ for all $i\in[m]$.
Let
$$
p = {\rm Proj}_{f_1,\ldots,f_m}(f_G) \eqdef f_G - \sum_{i=1}^m (f_G,f_i) f_i ,
$$
which is the projection of $f_G$, the extension of $f$ by zero,
onto the orthogonal complement of the 
span of $f_1,\ldots,f_m$.
Then
\begin{equation}\label{eq_p_minus_f_epsilon}
\| p - f_G \|_2 \le m \bigl( \# V_\psi) \delta \|f_G\|_2
\end{equation} 
\end{lemma}
\begin{proof}
The triangle inequality and $\|f_i\|_2=1$ for all $i$ implies
$$
\| p - f_G \|_2 \le \sum_{i=1}^m |(f_G,f_i)|\, \|f_i\|_2
\le m |(f_G,f_i)| ,
$$
so the result follows from \eqref{eq_f_t_ineq}.
\end{proof}

The methods of \cite{friedman_tillich_generalized} also require the
following type of results, which we gather in two lemmas.

\begin{lemma}\label{le_f_t_pushing_methods}
Let $\psi\subset T$ be two graphs.  Then there exists a 
$f\from V_{\psi}\to\reals$ that maximizes $\cR_{A_{\psi}}$ 
over all functions
$V_{\psi}\to\reals$;
by scaling $f$ one can assume $f\ge 0$ (i.e., $f(v)\ge 0$ for all
$v\in V_{\psi}$).
If $f_T$ denotes the extension by zero of $f$ to $T$, and setting
$\lambda=\cR_{A_\psi}(f)$, we have
$A_T f_T\ge \lambda f_T$ pointwise, i.e., 
\begin{equation}\label{eq_f_t_superharm}
\bigl(A_T f_T\bigr) (v)\ge \lambda f_T(v), 
\ \forall v\in V_T .
\end{equation} 
\end{lemma}
\begin{proof}
The existence of $f\ge 0$ is implied by the Perron-Frobenius theorem
and the symmetry of $A_\psi$ (or Theorem~2.3 of 
\cite{friedman_geometric_aspects} for a more general context, such
as $G$ infinite and weights on $V_G$ in the Rayleigh quotient) and it satisfies
$$
A_{\psi} f = \lambda f
$$
on $V_\psi$.  It follows that $(A_G f_G)(v) = \lambda f_G(v)$ on $V_\psi$;
for $v\notin V_\psi$ we have $f_G(v)=0$ and
$$
(A_G f_G)(v) \ge 0 
$$
since $f_G\ge 0$ everywhere; this implies
that $(A_G f_G)(v) \ge \lambda f_G(v)$ for $v\notin V_\psi$, and hence
we deduce \eqref{eq_f_t_superharm}.
\end{proof}

\begin{definition}\label{de_push_forward}
If $\pi\from T\to G$ is any covering map
(of possibly infinite graphs), for any finitely supported
function $g\from V_T\to\reals$ we define the {\em push forward of $f$
along $\pi$}, denoted
$\pi_* g$, to be the function $V_G\to\reals$ given by
\begin{equation}\label{eq_pi_lower_star}
(\pi_* g)(v) = \sum_{\pi(u)=v} g(u) .
\end{equation} 
\end{definition}
(We easily see that if 
$\pi^*f \eqdef f \circ \pi$
then $\pi_*$ is the adjoint of $\pi^*$ 
with respect to the inner products
\eqref{eq_inner_product} on $L^2(V_T)$ and $L^2(V_G)$ provided that
$\pi$ is finite-to-one or we restrict to
finitely supported functions; 
however this is not particularly important
here.)

\begin{lemma}\label{le_f_t_pushing_forward}
Let 
$\pi\from T\to G$ be a covering map (of possibly infinite graphs),
and $g\from V_T\to\reals$ a finitely supported function such that
$A_T g \ge \lambda g$ (pointwise, i.e., at each $v\in V_T$)
for some $\lambda\in\reals$.  Then
\begin{equation}\label{eq_f_t_superharm_push}
A_G (\pi_* g ) \ge \lambda\, \pi_* g 
\end{equation} 
(pointwise); furthermore if $g\ge 0$ and is nonzero, then
\begin{equation}\label{eq_f_t_Rayleigh_push}
\cR_{A_G}(\pi_* g) \ge \lambda.
\end{equation} 
\end{lemma}
\begin{proof}
It is easy to see that $\pi_* A_T = A_G \pi_*$ on finitely supported 
functions, so that applying $\pi_*$ on the left 
to \eqref{eq_f_t_superharm} we get \eqref{eq_f_t_superharm_push}.
Furthermore, if $g\ge 0$ and is nonzero, then $f=\pi_* g\ge 0$ and
$f$ is nonzero, and hence
\eqref{eq_pointwise_implication} implies that 
$\cR_{A_G}(\pi_* g) \ge \lambda$.
\end{proof}

\subsection{Spectral Results on Infinite Graphs}

Let us quote some standard spectral results about 
infinte graphs used in \cite{friedman_alon}.
If $G$ is a {\em locally finite graph}, meaning that $V_G,\Edir_G$ may
be infinite but the degree of each vertex is finite, 
we use $\|A_G\|$ to denote the norm of the
adjacency matrix, $A_G$, of $G$, viewed as operator
on $L^2(V_G)$ where each vertex in $V_G$ has measure $1$.
It is a standard result in operator theory that $\|A_G\|$ is also
the spectral radius of $G$, and that if $G$ is connected then
for every $v\in V_G$ we have
\begin{equation}\label{eq_buck}
\| A_G \| = \lim_{r\to\infty} c(v,2r)^{1/(2r)}
\end{equation} 
where $c(v,k)$ denotes the number of closed walks in $G$ from $v$ (to $v$)
of length $k$
(e.g., \cite{buck}, Proposition~3.2).
It follows from standard spectral theory 
(see, for example, Theorem~3.11 of \cite{friedman_alon}) that 
for any $\epsilon>0$ there is a finitely supported function 
$f\from V_G\to\reals$, i.e., where
$$
{\rm supp}(f) \eqdef \{ v \ | \  f(v)\ne 0 \}
$$
is finite, for which 
\begin{equation}\label{eq_epsilon_finite_supp}
\| A_G f \| \ge \bigl( \| A_G \| - \epsilon \bigr) \|f\|.
\end{equation}

\subsection{Relative Trees}

Now we review some facts from Section~3 of \cite{friedman_alon}, stated
in our context.
First we describe the analogue of ${\rm Tree}_d(\psi)$ of
Section~3.8 of \cite{friedman_alon}; we state it as a lemma.

\begin{lemma}\label{le_universal}
Let $B$ be a graph and $\psi_\Bg$ a connected \'etale $B$-graph.  Then
there exists a graph
$T={\rm Tree}_\Bg(\psi_\Bg)$ (on infinitely many vertices) that is
a {\em universal cover of $B$-graphs extending $\psi_\Bg$} in the following
sense:
if $G_\Bg$ is a covering $B$-graph (possibly on infinitely many
vertices) that contains $\psi_\Bg$
as a subgraph, then there exists a unique morphism
$T_\Bg\to G_\Bg$ that fixes $\psi_\Bg$.
The graph $T={\rm Tree}_\Bg(\psi_\Bg)$ is unique up to unique isomorphism.
\end{lemma}
This lemma is illustrated by 
an example in
Subsection~3.8 of \cite{friedman_alon} 
in Figure~\ref{fi_dtree} (which we reproduce here): 
in this example we work with graphs rather
than $B$-graphs, and ${\rm Tree}_d(\psi)$ is the case of
${\rm Tree}_\Bg(\psi_\Bg)$ where $B$ consists
of a single vertex of degree $d$ (in any combination of whole-loops and
half-loops) and where we otherwise forget the $B$-structure.
\begin{figure}[h]
    \centering
         \includegraphics[width=3.5in]{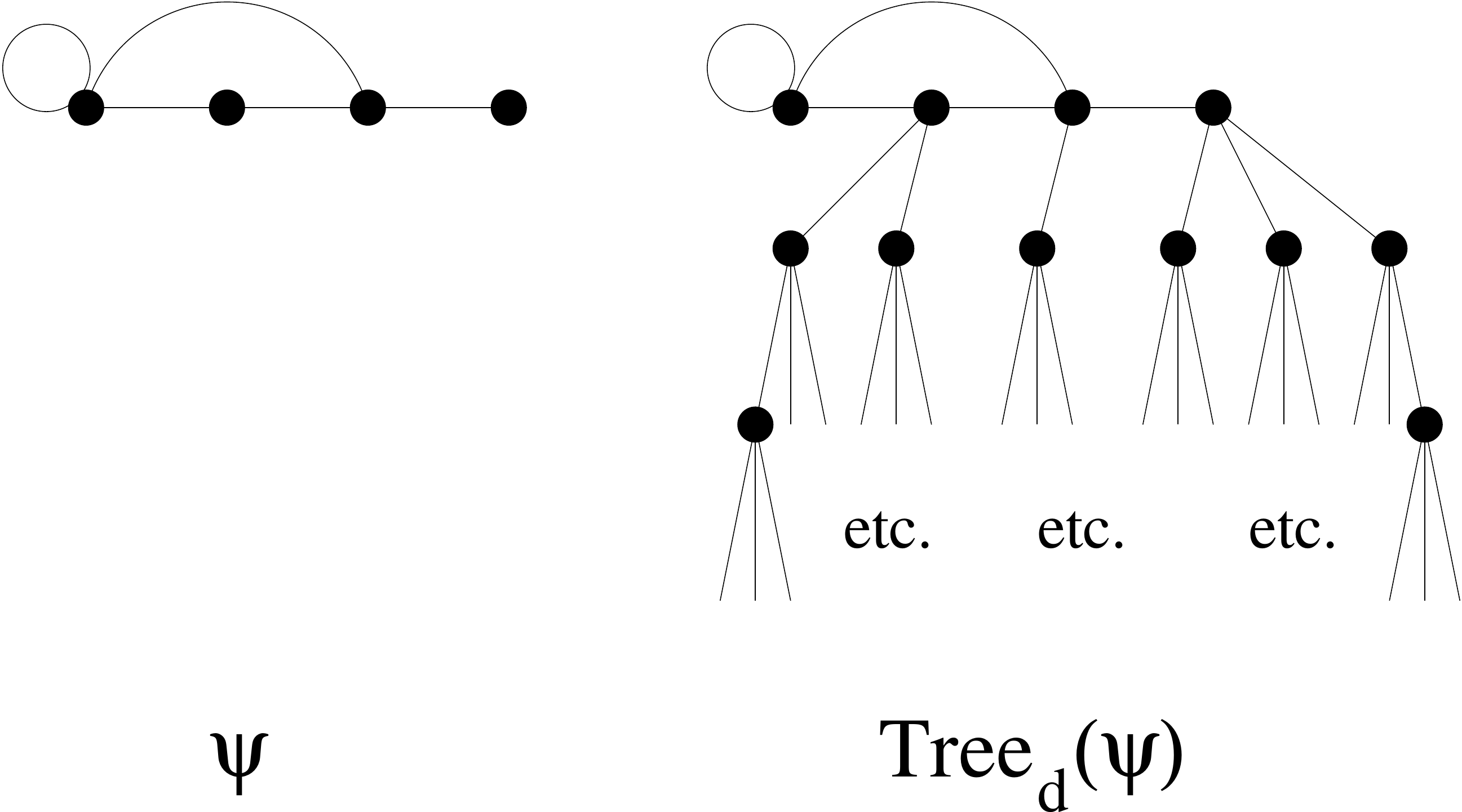}
    \caption{A graph, $\psi$, and ${\rm Tree}_d(\psi)$
        with $d=4$.}
    \label{fi_dtree}
\end{figure}
Let us describe the
proof of Lemma~\ref{le_universal} informally: the proof
is a standard adaptation of the notion of a universal cover:
we build ${\rm Tree}_\Bg(\psi_\Bg)$ by taking any vertex $v\in V_\psi$
where $\psi\to B$ is not a local isomorphism, and to each such $v$ we
add new edges, each with its own new vertex; this gives us a graph
$\psi_\Bg^1$ containing $\psi_\Bg$ such that $\psi_\Bg^1\to B$ is a local
isomorphism
at all vertices of $\psi_\Bg$; then $\psi_\Bg^1\to B$ can only fail to be a local
isomorphism at the newly created leaves of $\psi_\Bg^1$, and we similarly
create $\psi_\Bg^2$ by adding edges and leaves to $\psi_\Bg^1$;
we continue to build $\psi_\Bg^3,\psi_\Bg^4,\ldots$, with
$$
\psi_\Bg\subset \psi_\Bg^1 \subset \psi_\Bg^2\subset \cdots
$$
and we set
$T_\Bg={\rm Tree}_\Bg(\psi_\Bg)$ to be the limit (i.e., the union).  
Then we easily show that
$T_\Bg$ has the universal property claimed.

\begin{definition}
Given a graph, $B$ and a connected \'etale $B$-graph, $\psi_\Bg$, the 
{\em relative $\psi_\Bg$ tree} refers to any $B$-graph
isomorphic to ${\rm Tree}_\Bg(\psi_\Bg)$.
\end{definition}

\begin{lemma}\label{le_spectral_rad}
Let $B$ be a $d$-regular graph for some $d\ge 3$, and let 
$\psi_\Bg$ be a $B$-graph with $\mu_1(\psi)>(d-1)^{1/2}$.
Setting $T_\Bg={\rm Tree}_\Bg(\psi_\Bg)$, we have
$$
\|A_T\| = \mu_1(\psi) + (d-1)/\mu_1(\psi)
$$
where $A_T$ is the adjacency operator/matrix on the infinite graph, $T$
and $\|A_T\|$ denotes the norm of $A_T$ as an operator on
$L^2(V_T)$.
\end{lemma}
This lemma is a more precise version of Theorem~3.13 of \cite{friedman_alon},
and will be proven in a similar fashion;
the special case of this lemma where $\psi$ is a bouquet of whole
loops appears in \cite{puder}, Table 2 (in the column ``The Growth Rate
for the Bouquet $B_{d/2}$'').
For sake of completeness we review facts about Shannon's algorithm
proven in Section~3 of \cite{friedman_alon} so that we can precisely
quote the results in the proof of Theorem~3.13 there.

\subsection{Shannon's Algorithm}

To prove Lemma~\ref{le_spectral_rad} we need to recall some facts 
from Section~3.4 of \cite{friedman_alon}
regarding {\em Shannon's algorithm} to compute the Perron-Frobenius
$\lambda_1(G)$
where $G=\VLG(T,\mec k)$ for a directed graph, $T$ and
$\mec k\from\Edir_T\to\naturals$.

First, if $Z=Z(z)$ is a $n\times n$ matrix whose entries are formal
power series
in a single indeterminate $z$ with non-negative coefficients, then we 
easy verify that the following are equal under the assumption that
the radius of convergence of each entry of $Z(z)$ is positive:
\begin{enumerate}
\item the smallest positive real solution, $z_0>0$,
of the equation $\det(I-Z(z))=0$,
\item the smallest radius of convergence, $z_0>0$, among the radii of
convergence of the $n^2$ entries
of the matrix of power series
$$
I + Z(z) + Z^2(z) + \cdots 
$$
\end{enumerate}
(allowing for $z_0=\infty$, i.e., $\det(I-Z(z))=0$ has no positive
real solution and, equivalently, the above matrix power series
converging for all $z\in\complex$).

\begin{definition}
Let $Z=Z(z)$ be an square matrix whose entries are formal power series
in a single indeterminate $z$ with non-negative coefficients, and assume
that each entry has a nonzero radius of convergence.
We call $1/z_0$ the {\em valence of $Z(z)$}
(if $z_0=\infty$ we say the valence is zero).
\end{definition}

If $T$ is a directed graph, (possibly with a countable number of
vertices and directed edges) and $\mec k\from \Edir_T\to\naturals$,
we define a square matrix 
$Z_{T,\mec k}=Z_{T,\mec k}(z)$ indexed on
$V_T$ and whose entries are formal power series in an indeterminate $z$ whose
$v_1,v_2$ entry (with $v_1,v_2\in V_T$) given by
$$
Z_{v_1,v_2}=\bigl(Z_{T,\mec k}(z)\bigr)_{v_1,v_2} = 
\sum_{h(e)=v_1,t(e)=v_2} z^{k(e)}.
$$
Shannon noted that if $T$ is finite, then
$\lambda_1(\VLG(T,\mec k))$ is the valence of $Z_{T,\mec k}(z)$.

Shannon's algorithm is remarkably robust; for example, it generalizes well for
$\mec k\from\Edir_T\to\reals_{>0}$, and when $\Edir_T$ is infinite.
The
situation where $\Edir_T$ arises naturally, such as when one wishes to
``eliminates vertices'' in the same way and for the same reason that one
eliminates states in a GNFA to produce a regular expression from a finite
automaton (see \cite{sipser}, Chapter~1); it also arises when approximating
an infinite graph by an equivalent graph with a finite subset of vertices,
as we now explain.

\subsection{An Infinite Version of Shannon's Algorithm}

Let us formally state the infinite version of Shannon's algorithm
needed in \cite{friedman_alon}.

\begin{definition}\label{de_digraph_suppression}
Let $G$ be a possibly infinite directed graph, and $V'\subset V_G$.
By the {\em directed suppression of $V'$ in $G$}, denoted
$G/V'$, we mean the graph $H$ given by
\begin{enumerate}
\item $V_H=V_G\setminus V'$;
\item $\Edir_H$ is the set of (finite length) walks in $G$
$$
w=(v_0,e_1,\ldots,e_k,v_k)
$$
such that $v_0,v_k\in V_H$, and $v_1,\ldots,v_{k-1}\in V'$;
we set $h_H w=v_k$, $t_H w=v_0$.
\end{enumerate}
We define the {\em edge-lengths of $G/V'$} to be the
function $\Edir_{G/V'}\to\naturals$
taking $w\in\Edir_{G/V'}$ as above to its length, $k$.
\end{definition}
This notion of suppression is akin to the notion of bead suppression
used to define homotopy type; however, even if $G$ is finite, $G/V'$
will generally have infinitely many directed edges
(unless no element of $\Edir_H$ traverses
a vertex of $V'$ twice).

If $G$ above is the underlying directed graph of a graph with involution
$\iota_G$, then $G/V'$ becomes a graph with the evident involution
$\iota_H$ taking $w=(v_0,e_1,\ldots,e_k,v_k)$ to its reverse walk
$$
w^R \eqdef (v_k,\iota_G e_k,\ldots,\iota_G e_1,v_0) .
$$

\begin{definition}
For a digraph $G$, we define {\em edge-lengths on $G$} to be
any function $\bec\ell\from\Edir_G\to\naturals$.  
If $w=(v_0,e_1,\ldots,e_k,v_k)$ is a walk in $G$, we define
the {\em $\bec\ell$-length of $w$} to be
$$
\bec\ell(w) \eqdef \ell(e_1)+\cdots+\ell(e_k)
$$
(the notation $\bec\ell(w)$ is different from the values, $\ell(e)$,
of $\bec\ell$ on $e$ and is unlikely to cause confusion).
If, in Definition~\ref{de_digraph_suppression},
$\bec\ell$ is a set of edge-lengths on $G$, we define the
{\em restriction of $\bec\ell$ to $G/V'$}
to be the {\em edge-lengths} on $G/V'$,
i.e., the function $\bec\ell'\from\Edir_{G/V'}\to\naturals$,
taking $w=(v_0,e_1,\ldots,e_k,v_k)$ to
$$
\bec\ell'(w) \eqdef \bec\ell(w)=\ell(e_1) + \cdots + \ell(e_k).
$$
\end{definition}

It is easy to see that if $G$ is any connected digraph that is
the underlying digraph of some graph, and $G$ is endowed with edge-lengths
$\bec\ell\from\Edir_G\to\naturals$, 
and $V'\subset V_G$ is any proper subset,
then for any $v\in V_{G/V'}$ we have that
the $v,v$ entry of
$$
I+ Z_{G/V',\bec\ell'}(z) + Z^2_{G/V',\bec\ell'}(z) + \cdots
$$
is just the sum 
$$
\sum_{k\ge 0} c_{G,\bec\ell}(v,k) z^k
$$
where $c_{G,\bec\ell}(v,k)$, 
akin to \eqref{eq_buck},
denotes the number of walks, $w$ in $G$ from $v$ to
itself of length $k=\bec\ell(w)$.

If the degree of each vertex of $G$ is bounded above, and the values
of $\bec\ell$ are bounded above,
then it is easy to see that
this power series has the same radius of convergence as
$$
\sum_{k \ {\rm even}} c_{G,\bec\ell}(v,k) z^k ,
$$
and that this radius is independent of $v$ if $G$ is connected. 
Hence if $G$ is a digraph (with unit edge-lengths $\bec\ell=\mec 1$) 
that is the
underlying digraph of a graph,
then \eqref{eq_buck} implies that this radius of convergence equals
$1/\|A_G\|$.
Let us formally record this fact, since it is fundamental to our methods;
it is a sort of ``infinite version'' of Shannon's algorithm.

\begin{lemma}\label{le_shannon_infinite}
Let $G$ be a connected graph, possibly infinite, with vertex degree 
bounded from above.
Then if $V'\subset G$ is a proper subset, and $\mec k$ are the edge
lengths of $G/V'$ (i.e., induced from unit length on $\Edir_G$)
then the valence 
of $\VLG(G/V',\mec k)$ equals $\|A_G\|$.
\end{lemma}

\subsection{The Curious Theorem of \cite{friedman_alon}}

In this subsection we recall some results from the proof of
Theorem~3.13 (the ``Curious Theorem'') of \cite{friedman_alon}.

If $T_d$ is the infinite undirected 
rooted tree each of whose interior node has
$d-1$ children,
and for $k\in\naturals$
we use $a_k$ to denote the number of
walks of length $k$ from the root to itself (which is zero when $k$ is 
odd), then we set
$$
S_d(z) = \sum_{k=2}^\infty a_k z^k
$$
which is a power series with non-negative integer coefficients which we
easily see is given by (see \cite{friedman_alon}, equation (13) after
the statement of Theorem~3.13)
\begin{equation}\label{eq_S_d}
S_d(z)= \frac{1 - \sqrt{1-4(d-1)z^2} }{2},
\end{equation} 
in the sense that we understand
$$
\sqrt{1-4(d-1)z^2} = 1 - (1/2) 4(d-1)z^2 
+ \binom{1/2}{2} \bigl( 4(d-1)z^2 \bigr)^2
- \binom{1/2}{3} \bigl( 4(d-1)z^2 \bigr)^2 + \cdots
$$
so that $S_d(z)=(d-1)z^2+\cdots$ is a power series with
non-negative coefficients.
Of course, $S_d(z)/(d-1)$ represents the same series for a rooted
tree where the root has one child and each other node has $d-1$ 
children.

If $\psi_\Bg$ is a $B$-graph of a $d$-regular graph, then if $D_\psi$ is 
the diagonal degree counting of $\psi$, then
\begin{equation}\label{eq_Z_T}
Z(z) = Z_\psi(z) + \frac{S_d(z)}{d-1} (dI-D_\psi) , \quad
Z_\psi(z)=zA_\psi 
\end{equation} 
is the power series representing Shannon's algorithm matrix
$Z_\psi(z)=zA_\psi$ plus the addition at each $v\in V_\psi$ a
series for the number of walks from $v$ to $v$ along an additional
$d-\deg(v)$ edges each of which is grown into a $d$-regular tree.
Hence the power series in $z$ representing
all of walks in $T={\rm Tree}_\Bg(\psi_\Bg)$
from any vertex of $V_\psi$ to another is given by the corresponding
entry of
$$
I+Z(z)+\bigl( Z(z) \bigr)^2+ \cdots
$$
It follows from Theorem~3.10 of \cite{friedman_alon} (which refers to
\cite{buck} for a proof) that the valence of $Z(z)$ is the spectral radius of
$A_T$.
Then an easy computation shows that \eqref{eq_Z_T} implies that
\begin{equation}\label{eq_main_curious_result}
I - Z(z) = \bigl( 1 - S_d(z) \bigr) 
\bigl( I - y A_\psi + y^2(D_\psi-I) \bigr),
\quad\mbox{where}
\ y=y(z) = \frac{z}{1-S_d(z)}.
\end{equation} 

\subsection{Proof of Lemma~\ref{le_spectral_rad}}

\begin{proof}
By the Ihara Determinantal Formula, Theorem~\ref{th_Ihara}
(or \cite{godsil}, Exercise~13,
page~72), we have that $\mu_1(\psi)=1/y_0$ where $y_0>0$ is the smallest
positive root of 
$$
\det \bigl(I - y_0 A_\psi + y_0^2 (D_\psi-I) \bigr).
$$
Since $\mu_1(\psi)>(d-1)^{1/2}$, we have $y_0<1/(d-1))^{1/2}$.  According to
\eqref{eq_main_curious_result}, $\lambda_1(T_\Bg)=1/z_0$ where $z_0>0$
is the radius of convergence at $z=0$ of the power series
\begin{equation}\label{eq_rad_conv_we_need}
\bigl( I - Z(z) \bigr)^{-1}
=
\Bigl(
\bigl( 1 - S_d(z) \bigr) 
\bigl( I - y(z) A_\psi + y^2(z) (D_\psi-I) \bigr)  \Bigr)^{-1}
\end{equation} 
where 
$$
y=y(z) = \frac{z}{1-S_d(z)}.
$$
The radius of convergence $S_d(z)$ is easily seen to be
$z=1/(2\sqrt{d-1})$.  The radius of convergence of
$$
\bigl( I - y(z) A_\psi + y^2(z) (D_\psi-I) \bigr)^{-1}
$$
is $1/z_0$, where $z_0$ the smallest value of $z$ where 
$y(z)=y_0=1/\mu_1(\psi)$, i.e., with
$$
\frac{z}{1-S_d(z)} =  1/\mu_1 = 1/\mu_1(\psi).
$$
Solving for $z$ we get
$$
z\mu_1 = 1 - S_d(z) = 
\frac{1 + \sqrt{1-4(d-1)z^2}}{2}
$$
where $\mu_1=\mu_1(\psi)$, so
$$
2z\mu_1-1 = \sqrt{1-4(d-1)z^2}
$$
so 
$$
4z^2\mu_1^1 - 4z\mu_1 + 1 = 1 - 4(d-1)z^2
$$
so
$$
z(d-1+\mu_1^2) = \mu_1 .
$$
Hence
$$
1/z_0 = \frac{d-1}{\mu_1} + \mu_1
$$
which is no smaller than $2\sqrt{d-1}$.  Hence
the radius of convergence of \eqref{eq_rad_conv_we_need} is
this $z_0$, and hence 
$$
\lambda_1( T )  = \frac{d-1}{\mu_1}+ \mu_1
$$
where $\mu_1=\mu_1(\psi)$.
\end{proof}

\subsection{Proof of Lemma~\ref{le_fundamental_subgraph}}

\begin{proof}[Proof of Lemma~\ref{le_fundamental_subgraph}]
According to Lemma~\ref{le_spectral_rad},
$\lambda$ in \eqref{eq_lambda_fund} is just $\| A_T\|$.
Fix an $\epsilon>0$.
According to \eqref{eq_epsilon_finite_supp}
with $T_\Bg={\rm Tree}_\Bg(\psi_\Bg)$,
there is a finitely supported $g\ge 0$ in $L^2(V_T)$ such that
$$
\| A_T g\| \ge \bigl( \| A_T \| - \epsilon \bigr) \| g \|
=  \bigl( \lambda  - \epsilon \bigr) \| g \| .
$$
Let $\psi$ be the subgraph of $T$ induced on the vertices either
in support of $g$ or connected to the support by an edge.
Then $\psi$ is finite and contains both the support of $g$ and $A_T g$;
it follows that
$$
\| A_\psi g\|_{L^2(V_\psi)} \ge
\bigl( \lambda  - \epsilon \bigr) \| g \|_{L^2(V_\psi)} .
$$
It follows that if $f\ge 0$ is a the Rayleigh quotient maximizer on $\psi$,
then $\cR_{A_\psi}(f)\ge\cR_{A_\psi}(g)$, and
hence Lemma~\ref{le_f_t_pushing_methods} implies that
$$
A_T f_T  \ge (\lambda-\epsilon) f_T 
$$
(i.e., pointwise).

Now assume that 
$G_\Bg$ is a $B$-graph that contains a $B$-subgraph $\psi'_\Bg$
isomorphic to $\psi_\Bg$; there is a covering morphism
$\mu\from T\to G$ taking $\psi_\Bg$ to $\psi'_\Bg$.
Hence Lemma~\ref{le_f_t_pushing_methods} implies that
$$
A_G (\mu_* f_T) \ge (\lambda-\epsilon) \mu_* f_T,
$$
and hence 
$\mu_* f_T$ is a finitely supported function on $G$ and
(by \eqref{eq_pointwise_implication})
\begin{equation}\label{eq_mu_f_T}
\cR_{A_G}( \mu_* f_T ) \ge \lambda-\epsilon.
\end{equation} 

Now for $v\in V_B$ let $\II_v\from V_B\to\{0,1\}$ be the Dirac delta
function (i.e., indicator function) 
of $v$.
Then
for any covering map $\pi\from G\to B$ of degree $n$, for $v$ varying
over $V_B$,
$$
f_v \eqdef \II_v\circ\pi/\sqrt{n}
$$
are orthonormal functions on $V_G$, and $\| f\|_\infty = 1/\sqrt{n}$.
Of course, the projection of any function onto the subspace orthogonal 
to all the $f_v$ is a new function on $V_G$ with respect to $\pi\from G\to B$;
let $p$ be this projection applied to $\mu_* f_T$.
In view of Lemma~\ref{le_friedman_tillich} we have 
that if $m (\# V_\psi)/\sqrt{n}\le \epsilon$ then
$$
\| p - \mu_* f_T \|_2 \le  \epsilon \| \mu_* f_T \|_2,
$$
whereupon Lemma~\ref{le_Rayleigh_perturb} implies that
$$
\hbox{\vbox{
\halign{\hfil $#$ & \hfil $#$ \hfil & $#$ \hfil \cr
\cR_{A_G}(p) \ge & \cR_{A_G}(\mu_* f_T) & - \,\epsilon \,\|A_G\|_2  \cr
 \ge &  \lambda-\epsilon & - \,\epsilon \,d \cr
 =   & \omit\span \lambda - \epsilon(1+d) \cr
}}}
$$
Hence some new function has $A_G$ Rayleigh quotient at least
$\lambda-\epsilon(1+d)$, and hence the largest eigenvalue of $A_G$
restricted to the new functions is at least
$\lambda-\epsilon(1+d)$.
Replacing $\epsilon$ by $\epsilon/(1+d)$ we conclude that
$A_G$ has 
a new eigenvalue of at least $\lambda-\epsilon$ for $n$ sufficiently large.
\end{proof}


\section{Proofs of Theorems~\ref{th_hastangles_lower_bound}
and \ref{th_rel_Alon_regular}}
\label{se_remaining_proofs}

In this section we gather the results of 
Sections~\ref{se_proof_first} and \ref{se_fundamental_subgr}
to prove
Theorems~\ref{th_hastangles_lower_bound}
and \ref{th_rel_Alon_regular}.

\begin{proof}[Proof of Theorem~\ref{th_hastangles_lower_bound}]
Since $S$ occurs in $\cC_n(B)$, some $G\in\cC_n(B)$ for some $n$
contains be as a subgraph, and the $B$-graph structure on $S$
endows $S$ with the structure of a $B$-graph, $S_\Bg$, that
occurs in $\cC_n(B)$.  
Therefore $S_\Bg$ is a $(\ge\nu,<r)$ tangle, which occurs in $\cC_n(B)$,
and hence (Theorem~\ref{th_extra_needed})
\begin{equation}\label{eq_S_in_G_prob}
\Prob_{G\in\cC_n(B)}\Bigl[ [S_\Bg]\cap G\ne\emptyset 
\Bigr] \ge
C' n^{-\tau_{\rm tang}}  .
\end{equation} 
Applying Lemma~\ref{le_fundamental_subgraph} with $\psi_B=S_B$ we
have that for any $\epsilon>0$, for
$n$ sufficiently large, 
$[S_\Bg]\cap G_\Bg\ne\emptyset$ (i.e., if $G_\Bg$ has a subgraph isomorphic to
$S_\Bg$), then $A_G$ has a new eigenvalue at 
least 
$$
\mu_1(S) + \frac{d-1}{\mu_1(S)} - \epsilon = 2(d-1)^{1/2}+\epsilon_0-\epsilon.
$$
Taking $\epsilon=\epsilon_0/2$ we have that for some $n_0$,
\begin{equation}\label{eq_S_in_G_implies_nonAlon}
[S_\Bg]\cap G\ne\emptyset \quad\implies\quad
\bigl({\rm NonAlon}_B(G;\epsilon_0/2)>0 
\end{equation} 
provided that $n\ge n_0$.

In view of 
\eqref{eq_r_nu_for_hastangles_lower_bound},
$S$ is a $(\ge\nu,<r)$-tangle, and therefore we have
\begin{equation}\label{eq_S_in_G_implies_hastangles}
[S_\Bg]\cap G\ne\emptyset \quad\implies\quad
G \in {\rm HasTangles}(\ge\nu,<r) .
\end{equation} 
Combining
\eqref{eq_S_in_G_prob}--\eqref{eq_S_in_G_implies_hastangles}
implies
\eqref{eq_Has_Tangles_non_Alon_lower}.

To prove \eqref{eq_Has_Tangles_non_Alon_upper}, we see that
the results on the $\widetilde c_i$ in the asymptotic expression for 
\eqref{eq_main_tech_result2}
in Theorem~\ref{th_main_tech_result} implies that
that for every $\nu$ and
$r\in\naturals$ we have
$$
\Prob_{G\in\cC_n(B)}[ 
G\in {\rm HasTangles}(\ge\nu,<r)
 ] \le
C(\nu,r) n^{-j}
$$
where $j$ is the smallest order of any $(\ge\nu,<r)$ tangle;
in view of 
\eqref{eq_r_nu_for_hastangles_lower_bound}, we have 
$j=\tau_{\rm tang}$.  
This upper bound is also valid for any sub-event of $G$ of the event that
$G\in {\rm HasTangles}(\ge\nu,<r)$, and therefore we have
\eqref{eq_Has_Tangles_non_Alon_upper}.
\end{proof}

\begin{proof}[Proof of Theorem~\ref{th_rel_Alon_regular}]
For any $\nu,r,\epsilon,n$ we have
$$
f_0(n,\epsilon)
=\Prob_{G\in\cC_n(B)}\Bigl[{\rm NonAlon}_B(G;\epsilon)>0 \Bigr]
$$
is the sum of
$$
f_1(\nu,r,n,\epsilon) \eqdef \Prob_{G\in\cC_n(B)}\Bigl[
\bigl( G \in {\rm HasTangles}(\ge\nu,<r) \bigr)
\ \mbox{\rm and}\  %
\bigl({\rm NonAlon}_B(G;\epsilon)>0  \bigr)
\Bigr]
$$
and
$$
f_2(\nu,r,n,\epsilon) \eqdef \Prob_{G\in\cC_n(B)}\Bigl[
\bigl( G \in {\rm TangleFree}(\ge\nu,<r) \bigr)
\ \mbox{\rm and}\  %
\bigl({\rm NonAlon}_B(G;\epsilon)>0  \bigr)
\Bigr]
$$

Let us first show that for $\epsilon>0$ sufficiently small there
is a constant $C=C(\epsilon)$ such that
\begin{equation}\label{eq_tau_1_bound}
f_0(n,\epsilon) \le C(\epsilon) n^{-\tau_1}.
\end{equation} 
By definition of $\tau_{\rm tang}$, there exists a graph
$S$ occurring in $\cC_n(B)$ with
$\ord(S)=\tau_{\rm tang}$ and $\mu_1(S)>(d-1)^{1/2}$.
Fix any such $S$.
By Theorem~\ref{th_hastangles_lower_bound} we have
\begin{equation}\label{eq_upper_bound_f_1}
f_1(\nu,r,n,\epsilon) \le C(\nu,r)n^{-\tau_{\rm tang}}
\le C(\nu,r)n^{\tau_1},
\end{equation} 
where $C=C(\nu,r)$ is independent of $\epsilon$.
Next let $\nu_0>(d-1)^{1/2}$ be sufficiently small, and 
$r_0\in \naturals$ be sufficiently large so that
$\tau_{\rm alg}(\nu,r)\ge\tau_1$:
this is possible if $\tau_{\rm alg}$ is finite, by the paragraph after
below Definition~\ref{de_algebraic_power},
and also possible if $\tau_{\rm alg}=+\infty$ by similar observations,
since in this case $\tau_1=\tau_{\rm tang}$ is finite.
Let $\epsilon_0$ be given by
$$
2(d-1)^{1/2}+ \epsilon_0 = \nu_0 + \frac{d-1}{\nu_0}.
$$
If $0<\epsilon/2<\epsilon_0$, then we have
$$
2(d-1)^{1/2}+ \epsilon/2 = \nu' + \frac{d-1}{\nu'},
$$
with $(d-1)^{1/2}<\nu'<\nu$; applying
apply Theorem~\ref{th_rel_Alon_regular2} with
$\epsilon'$ set to $\epsilon/2$, and where $\epsilon$ in
the theorem is taken to be some number $\le \epsilon/2$;
it follows that for
$\tilde\epsilon$ between $\epsilon'$ and $\epsilon'+\epsilon/2$ we have
\begin{equation}\label{eq_upper_bound_f_2}
f_2(\nu,r,n,\tilde \epsilon) \le C n^{-\tau_1}.
\end{equation} 
Adding  
\eqref{eq_upper_bound_f_1} and \eqref{eq_upper_bound_f_2} we have
that for each $\epsilon>0$ sufficiently small,
$$
f_0(n,\tilde\epsilon)=
f_1(\nu,r,n,\tilde \epsilon) +
f_2(\nu,r,n,\tilde \epsilon)  \le C'' n^{-\tau_1},
$$
for some $\tilde\epsilon\le\epsilon$ and $C''$.
Since $f_0(n,\epsilon)$ is clearly non-increasing in $\epsilon$, we have
$$
f_0(n,\epsilon)\le f_0(n,\tilde\epsilon) \le C'' n^{-\tau_1} .
$$
This proves \eqref{eq_tau_1_bound}.

Next let us show that there is a constant, $C'$, such that
$\epsilon>0$ sufficiently small we have
\begin{equation}\label{eq_tau_2_bound}
f_0(n,\epsilon) \ge C' n^{-\tau_2}.
\end{equation} 
First consider the case where $\tau_{\rm tang}=\tau_2$.

In this case, fix a graph
$S$ occurring in $\cC_n(B)$ with
$\ord(S)=\tau_{\rm tang}$ and $\mu_1(S)>(d-1)^{1/2}$.
Let $\epsilon_0$ be given by
\eqref{eq_epsilon_0_S}, and consider any real $\nu$ and
$r\in\integers$ such that
\eqref{eq_r_nu_for_hastangles_lower_bound} holds.
According to Theorem~\ref{th_hastangles_lower_bound},
for $n$ sufficiently large we have
$$
f_1(\nu,r,n,\epsilon_0/2) \ge C' n^{-\tau_{\rm tang}} = C' n^{-\tau_2}.
$$
Then
$$
f_0(n,\epsilon_0/2)\ge f_1(\nu,r,n,\epsilon_0/2) \ge C' n^{-\tau_2},
$$
and hence for any $\epsilon\le\epsilon_0/2$ we have
$$
f_0(n,\epsilon)\ge f_0(n,\epsilon/2) \ge C' n^{-\tau_2},
$$
which proves
\eqref{eq_tau_2_bound}.  

Next consider the case that $\tau_{\rm tang}\ne\tau_2$.
In this case 
$\tau_2= \tau_{\rm alg}+1$.
Then, as above, by the paragraph after
below Definition~\ref{de_algebraic_power},
there are $\nu>(d-1)^{1/2}$ and $r\in\naturals$ such that
$\tau_{\rm alg}(\nu,r)$ (as in Theorem~\ref{th_rel_Alon_regular2})
equals $\tau_{\rm alg}$.  Then 
\eqref{eq_rel_alon_expect_lower_and_upper} implies that
$$
\EE_{G\in\cC_n(B)}[ \II_{{\rm TangleFree}(\ge\nu,<r)}(G)
{\rm NonAlon}_d(G;\epsilon_0/2+\epsilon) ] \ge C' n^{-\tau_{\rm alg}}
$$
for $\epsilon>0$ sufficiently small.  Setting 
$\epsilon_1=\epsilon_0/2+\epsilon$, we have
$$
\EE_{G\in\cC_n(B)}[{\rm NonAlon}_d(G;\epsilon_1) ] \ge
\EE_{G\in\cC_n(B)}[ \II_{{\rm TangleFree}(\ge\nu,<r)}(G)
{\rm NonAlon}_d(G;\epsilon_1) ] \ge C' n^{-\tau_{\rm alg}}
$$
for some $\epsilon_1>\epsilon_0/2$.
Since the number of new eigenvalues of $A_G$ with $G\in\cC_n(B)$
is $(n-1)(\#V_B)$,  we have
$$
(n-1)(\#V_B) f_2(\nu,r,n,\epsilon_1) \ge 
\EE_{G\in\cC_n(B)}[{\rm NonAlon}_d(G;\epsilon_1) ] \ge C' n^{-\tau_{\rm alg}},
$$
and hence for $n>1$ there is a constant $C''$ such that
$$
f_2(\nu,r,n,\epsilon_1) \ge
C'' n^{-\tau_{\rm alg}-1} = C'' n^{-\tau_2}.
$$
It follows that for any $\epsilon<\epsilon_1$ we have
$$
f_2(n,\epsilon) \ge f_2(\nu,r,n,\epsilon_1) \ge C'' n^{-\tau_2}.
$$
This establishes \eqref{eq_tau_2_bound} in the case where
$\tau_2= \tau_{\rm alg}+1$, and hence \eqref{eq_tau_2_bound} in
both cases.

\end{proof}

\section{Improved Markov Bounds in Trace Methods}
\label{se_markov_bounds}

The papers
\cite{broder,friedman_random_graphs,friedman_relative,linial_puder,puder}
get spectral bounds by applying
a Markov type bound to the expected trace of a single power 
of the adjacency matrix of the random graphs.
[This contrasts with the Sidestepping Theorem, whose proof involves
the expected trace of a number of consecutive powers of random
matrices; see Article~IV.]
It seems to have gone unnoticed until \cite{friedman_kohler}
that all these papers
get better results for regular graphs
by working with expected Hashimoto traces as opposed to adjacency traces.
In this section we discuss these improvements. 
The formulas we give demonstrate these improvements, although we have
no principle that explains why the Hashimoto matrix approach gives
better bounds (or if there could be further improvement using some
other ``transform'' of the adjacency matrix into equivalent spectral
information).

%

In this section we use $\rhonew_B(A_G)$ to denote the spectral
radius of $A_G$ 
restricted to the new functions of
a covering map $G\to B$, i.e., the largest absolute value of 
a new eigenvalue.
Since the sum of the $k$-th powers of the eigenvalues of $A_G$ restricted
to these new functions equals $\Trace(A_G^k) - \Trace(A_B^k)$, we see
that the new eigenvalues of $A_G$, and therefore $\rhonew_B(A_G)$, depends
only on $G$ and $B$ and not on the particular covering map $G\to B$.
We similarly use $\rhonew_B(H_G)$ for the new Hashimoto spectral radius, 
and similar remarks apply.

\subsection{The Adjacency Markov-Type Bound}

The papers \cite{broder,friedman_random_graphs,friedman_relative} 
essentially give an estimate for some fixed $r$ and $\log k \ll \log n$
of the form
\begin{equation}\label{eq_A_G_trace_estimate}
\EE_{G\in\cC_n(B)}[ \Trace(A_G^k) - \Trace(A_B^k)] \le
\bigl(n+O(1)\bigr) \rho(A_{\hat B})^k+ O(1/n^r)k^{O(1)} \lambda_1(B)^k
\end{equation} 
where $\rho(A_{\hat B})$ denotes the spectral radius of the adjacency
operator on $\hat B$, the universal cover of $B$.
Choosing $k$ to be even and to balance the two summands of the
right-hand-side of \eqref{eq_A_G_trace_estimate} yields a 
high probability bound of
\begin{equation}\label{eq_A_G_markov}
\rhonew_B(A_G) \le \lambda_1(B)^{1/r}\rho(A_{\hat B})^{r/(r+1)}+\epsilon
\end{equation} 
for any $\epsilon>0$.

The papers \cite{broder,friedman_random_graphs}
do this for $B$ equal
to the bouquet of $d/2$ whole-loops and achieve
\eqref{eq_A_G_trace_estimate} for, respectively, $r=1$ and 
all $r$ with $2r-1< (d-1)^{1/2}$
(the bound in \cite{broder} is slightly weaker since their order $n$
term in \eqref{eq_A_G_trace_estimate} is larger).
For $B$ regular we have
$$
\lambda_1(B) = d ,\quad \rho(A_{\hat B}) = 2\sqrt{d-1};
$$
for $r=2\lfloor (d-1)^{1/2}+1\rfloor$ as in \cite{friedman_random_graphs},
the high probability bound \eqref{eq_A_G_markov} becomes
$$
2\sqrt{d-1} \Bigl( d/2\sqrt{d-1}\Bigr)^{1/r}+\epsilon
=
2\sqrt{d-1} + (1/2)\log_e d - \log_e(2) + o(1)
$$
for large $d$.

\subsection{Markov Hashimoto Bounds}

The papers \cite{broder,friedman_random_graphs} obtain
bounds \eqref{eq_A_G_trace_estimate} by first estimating expected counts
of non-backtracking walks.  
These methods can be restricted to SNBC walks, giving a bound
\begin{equation}\label{eq_H_G_trace_estimate}
\EE_{G\in\cC_n(B)}[ \Trace(H_G^k) - \Trace(H_B^k) ] \le
(d-1)^{k/2} O(1) + O(1/n^r)Ck^C(d-1)^k
\end{equation} 
for the same values of $r$ as they do in
\eqref{eq_A_G_trace_estimate}.  Taking $k$ even and using the fact that
all nonreal eigenvalues of $H_G$ have absolute value $(d-1)^{1/2}$, one
gets a high probability bound on the largest new eigenvalue of $H_G$ of
$$
\rhonew_B(H_G) \le 
(d-1)^{1/2} (d-1)^{1/(2r)} + \epsilon ,
$$
which in view of the Ihara Determinantal Formula
corresponds to a high
probability eigenvalue bound of
\begin{equation}\label{eq_A_G_estimate_via_H_G}
\rhonew_B(A_G) \le 
(d-1)^{1/2} \bigl( (d-1)^{1/(2r)}+(d-1)^{-1/(2r)} \bigr) + \epsilon
\end{equation}
Curiously this bound, which is based on seemingly equivalent walk estimates,
improves the bound \eqref{eq_A_G_markov}.
For example, for $r$ as in \cite{friedman_random_graphs},
\eqref{eq_A_G_estimate_via_H_G} yields a bound of
$$
2\sqrt{d-1} + \frac{(\log_e d)^2}{16 \, d^{1/2}} (1+o(1))
$$
which for large $d$ is a significant improvement.

\subsection{Improvements to \cite{puder}}

Doron Puder and the first author \cite{friedman_puder}
have noted a similar improvement in
the new $A_G$ eigenvalue bound for $d$-regular $G$ in \cite{puder},
again by first converting the expected adjacency trace bounds there to
expected Hashimoto trace bounds as above.
The bounds in \cite{puder} 
are of the form
$$
\EE_{G\in\cC_n(B)}[ \Trace(A_G^k) - \Trace(A_B^k)] \le
n c_{-1}(k) + c_0(k) + \ldots + c_{d-1}(k)/n^{d-1} ,
$$
where each $c_i(k)$ is bounded by roughly $\rho_i^k$ where
$$
2\sqrt{d-1}=\rho_{-1}\le \rho_0\le \cdots\le \rho_{d-1} = d
$$
(see \cite{puder}, Table~2).  Taking $k$ even and to balance the two
dominant terms of this estimate gives a new eigenvalue bound of
$$
2\sqrt{d-1}+ (0.86\ldots) + \epsilon.
$$
On the other hand, converting to expected Hashimoto traces gives the
improved bound \cite{friedman_puder} of
$$
2\sqrt{d-1} + O(d^{-1/2}).
$$


\providecommand{\bysame}{\leavevmode\hbox to3em{\hrulefill}\thinspace}
\providecommand{\MR}{\relax\ifhmode\unskip\space\fi MR }
\providecommand{\MRhref}[2]{%
  \href{http://www.ams.org/mathscinet-getitem?mr=#1}{#2}
}
\providecommand{\href}[2]{#2}

\end{document}